\documentclass[11pt]{llncs}
\usepackage{amsmath,amssymb,times,a4wide}
\usepackage[final]{graphicx}
\pdfoutput=1

\def\Prob{\mathrm{I\!Pr}}

\def\union{\cup}
\def\cut{\cap}
\def\eps{\varepsilon}

\def\NP{\mathbf{NP}}

\newcommand{\reals}{\mathbb{R}}

\def\IE{\mathrm{IE}}
\def\RIE{R_{\mbox{\fontsize{7pt}{7pt}\selectfont IE}}}

\def\SDP{\mbox{SDP-IE}}
\def\sgn{\mathrm{sgn}}

\def\insertfig#1#2#3{%
\begin{figure}[t]
\centerline{\includegraphics[width=#1]{#2}}
\caption{#3}
\end{figure}}

\title{On the Efficiency of Influence-and-Exploit Strategies\\
for Revenue Maximization under Positive Externalities%
\thanks{Research partially supported by an NTUA Basic Research Grant (PEBE 2009).}}

\author{Dimitris Fotakis \and Paris Siminelakis}

\institute{%
School of Electrical and Computer Engineering,\\
National Technical University of Athens, 157 80 Athens, Greece.\\
Email: {\tt fotakis@cs.ntua.gr},\ \ {\tt psiminelakis@gmail.com}}

\begin{document}

\maketitle

\begin{abstract}
We study the problem of revenue maximization in the marketing model for social networks introduced by (Hartline, Mirrokni, Sundararajan, WWW '08). In this setting, a digital product is sold to a set of potential buyers under positive externalities, and the seller seeks for a marketing strategy, namely an ordering in which he approaches the buyers and the prices offered to them, that maximizes his revenue. We restrict our attention to the Uniform Additive Model
%
%
and mostly focus on Influence-and-Exploit (IE) marketing strategies. We obtain a comprehensive collection of results on the efficiency and the approximability of IE strategies, which also imply a significant improvement on the best known approximation ratios for revenue maximization.
Specifically, we show that in the Uniform Additive Model, both computing the optimal marketing strategy and computing the best IE strategy are $\NP$-hard for undirected social networks.
We observe that allowing IE strategies to offer prices smaller than the myopic price in the exploit step leads to a measurable improvement on their performance.
Thus, we show that the best IE strategy approximates the maximum revenue within a factor of $0.911$ for undirected and of roughly $0.553$ for directed networks.
Moreover, we present a natural generalization of IE strategies, with more than two pricing classes, and show that they approximate the maximum revenue within a factor of roughly $0.7$ for undirected and of roughly $0.35$ for directed networks.
Utilizing a connection between good IE strategies and large cuts in the underlying social network, we obtain polynomial-time algorithms that approximate the revenue of the best IE strategy within a factor of roughly $0.9$.
Hence, we significantly improve on the best known approximation ratio for revenue maximization to $0.8229$ for undirected and to $0.5011$ for directed networks (from $2/3$ and $1/3$, respectively, by Hartline et al.).
\end{abstract}

\thispagestyle{empty}%
\pagenumbering{arabic}%
\setcounter{page}{0}%

\newpage
\pagestyle{plain}
\pagenumbering{arabic}

\section{Introduction}
\label{s:intro}

An important artifact of the Internet is the appearance of Social Networks. For the first time, we posses an account of the friendship network for millions of people. This information is invaluable for targeted advertising, personalized recommendations, and in general, for enhancing business intelligence. However, there is a noticeable discrepancy between the perceived value of Social Networks and the actual revenue they generate. For instance, Facebook is valued by Goldman Sachs at \$50 billions, but its actual revenue is estimated at \$2.2 billions by eMarketer. The widespread belief is that much of the commercial potential of Social Networks remains unexploited.

This premise has motivated a significant volume of research in the direction of monetizing Social Networks. Recent research has studied the impact of externalities in a variety of settings (see e.g. \cite{Hart,Arthur,Iterative,Anari,Chen,Bimpikis,HIMM11,FKT11}). In this work, we are interested in the design of efficient marketing strategies that exploit positive externalities and maximize the seller's revenue. We focus on the setting where the utility of the product depends inherently on the scale of the product's adoption. E.g. the value of a social network depends on the fraction of the population that uses it on a regular basis. In fact, for many products (e.g., cell phones, online gaming), the value of the product to a buyer depends on the set of his friends that have already adopted the product. In this setting, the seller seeks for a marketing strategy that guarantees a significant revenue through a wide adoption of the product, which leads to an increased value, and consequently, to a profitable pricing of it.

\smallskip\noindent{\bf A Model of Marketing Strategies for Social Networks.}
We adopt the model of Hartline, Mirrokni, and Sundararajan \cite{Hart}, where a digital product is sold to a set of potential buyers under positive externalities. As in \cite{Hart}, we assume an unlimited supply of the product and that there is no production cost for it. A (possibly directed) social network $G$ on the set of potential buyers models how their value of the product is affected by other buyers who already own the product.
Namely, an edge $e = (j, i)$ in the social network denotes that the event that $j$ owns the product has a positive influence on $i$'s value of the product. The product's value to each buyer $i$ is given by a non-decreasing function $v_i(S)$ of the set $S$ of buyers who own the product and have positive influence on (equivalently, an edge to) $i$. The exact values $v_i(S)$ are unknown and are treated as random variables of which only the distributions $F_{i,S}$ are known to the seller.
The most interesting cases outlined in \cite{Hart} are the \emph{Concave Graph Model}, where each $v_i(S)$ is a submodular function of $S$, and the \emph{Uniform Additive Model}, where a non-negative weight $w_{ji}$ is associated with each edge $(j, i)$, and $v_i(S)$ is uniformly distributed between $0$ and the total weight of the edges from $S$ to $i$. An important special case of the Uniform Additive Model is the undirected (or the symmetric) case, where the network $G$ is undirected and $w_{ij} = w_{ji}$ for all edges $\{ i, j \}$.

In this setting, the seller approaches each potential buyer once and makes an individualized offer to him. A \emph{marketing strategy} determines the sequence in which the seller approaches the buyers, and the price offered to each buyer. Each buyer either accepts the offer, in which case he pays the price asked by the seller, or rejects it, in which case he pays nothing and never receives an offer again. The seller's goal is to compute a marketing strategy that maximizes his revenue, namely the total amount paid by the buyers accepting the offer.

Using a transformation from Maximum Acyclic Subgraph, Hartline et al. \cite{Hart} proved that if the seller has complete knowledge of the buyers' valuations, computing a revenue-maximizing ordering of the buyers is $\NP$-hard for directed social networks. Combined with the result of \cite{GMR08}, this transformation suggests an upper bound of $0.5$ on the approximability of revenue maximization for directed networks and deterministic additive valuations. On the positive side, Hartline et al. gave a polynomial-time dynamic programming algorithm for a special fully symmetric case, where the order in which the seller approaches the buyers is insignificant.

An interesting contribution of \cite{Hart} is a class of elegant marketing strategies called \emph{Influence-and-Exploit}. An Influence-and-Exploit (IE) strategy first offers the product for free to a selected subset of buyers, aiming to increase the value of the product to the remaining buyers (influence step). Then, in the exploit step, it approaches the remaining buyers, in a random order, and offers them the product at the so-called \emph{myopic price}. The myopic price ignores the current buyer's influence on the subsequent buyers and maximizes the expected revenue extracted from him. In the Uniform Additive Model, the myopic price is accepted by each buyer with probability $1/2$. Hence, there is a notion of uniformity in the prices offered in the exploit step, in the sense that each buyer accepts the offer with a fixed probability, and we can say that the IE strategy uses a \emph{pricing probability} of $1/2$.

To demonstrate the efficiency of IE strategies, Hartline et al. \cite{Hart} proved that the best IE strategy approximates the maximum revenue within a factor of $0.25$ for the Concave Graph Model, and within a factor of $0.94$ for the (polynomially solvable) fully symmetric case of the Uniform Additive Model. Furthermore, they proved that if each buyer is selected in the influence set randomly, with an appropriate probability, the expected revenue of IE is at least $2/3$ (resp. $1/3$) times the maximum revenue of undirected (resp. directed) social networks. For the Concave Graph Model, Hartline et al. presented a polynomial-time local search algorithm which approximates the revenue of the best IE strategy within a factor of $0.4$.
%
%
Since \cite{Hart}, the Influence-and-Exploit paradigm has been applied to a few other settings involving revenue maximization under positive externalities (see e.g. \cite{Arthur,Bimpikis,HIMM11}).

\smallskip\noindent{\bf Contribution and Techniques.}
Despite the fact that IE strategies are simple, elegant, and quite promising in terms of efficiency, their performance against the maximum revenue and their polynomial-time approximability are hardly well understood. In this work, we restrict our attention to the important case of the Uniform Additive Model, and obtain a comprehensive collection of results on the efficiency and the approximability of IE strategies. Our results also imply a significant improvement on the best known approximation ratio for revenue maximization in the Uniform Additive Model.

We first show that in the Uniform Additive Model, both computing the optimal marketing strategy and computing the best IE strategy are $\NP$-hard for undirected social networks%
\footnote{We should highlight that if the seller has complete knowledge of the buyers' valuations, finding a revenue-maximizing buyer ordering for undirected social networks is polynomially solvable (cf. Lemma~\ref{l:ordering}). Therefore, the reduction of \cite{Hart} does not imply the $\NP$-hardness of revenue maximization for undirected networks.}.
Next, we embark on a systematic study of the algorithmic properties of IE strategies (cf. Section~\ref{s:ie}). In \cite{Hart}, IE strategies are restricted, by definition, to the myopic pricing probability, which for the Uniform Additive Model is $1/2$. A bit surprisingly, we observe that we can achieve a measurable improvement on the efficiency of IE strategies if we use smaller prices (equivalently, a larger pricing probability) in the exploit step. Thus, we let IE strategies use a carefully selected pricing probability $p \in [1/2, 1)$.
 %

We show the existence of an IE strategy with pricing probability $0.586$ (resp. $2/3$) which approximates the maximum revenue within a factor of $0.9111$ for undirected (resp. $0.55289$ for directed) networks.
The proof assumes a revenue-maximizing pricing probability vector $\vec{p}$ and constructs an IE strategy with the desired expected revenue by applying randomized rounding to $\vec{p}$. An interesting consequence is that the upper bound of $0.5$ on the approxibability of the maximum revenue of directed networks does not apply to the Uniform Additive Model. In Section~\ref{s:ie}, we discuss the technical reasons behind this and show a pair of upper bounds on the approximability of the maximum revenue of directed social networks in the Uniform Additive Model. Specifically, assuming the Unique Games conjecture, we show that it is $\NP$-hard to approximate the maximum revenue within a factor greater than $27/32$, and greater than $3/4$, if we use an IE strategy with pricing probability $2/3$.

The technical intuition behind most of our results comes from the apparent connection between good IE strategies and large cuts in the underlying social network. Following this intuition, we optimize the parameters of the IE strategy of \cite{Hart} and slightly improve the approximation ratio to $0.686$ (resp. $0.343$) for undirected (resp. directed) social networks. Moreover, we show that for undirected bipartite social networks, an IE strategy extracts the maximum revenue and can be computed in polynomial time.
Building on the idea of generating revenue from large cuts in the network, we discuss, in Section~\ref{s:ie-gen}, a natural generalization of IE strategies that use more than two pricing classes. We show that these strategies approximate the maximum revenue within a factor of $0.7032$ for undirected networks and of $0.3516$ for directed networks.

The main hurdle in obtaining better approximation guarantees for the maximum revenue problem is the lack of any strong upper bounds on it. In Section~\ref{s:sdp}, we obtain a strong Semidefinite Programming (SDP) relaxation for the problem of computing the best IE strategy with any given pricing probability. Our approach exploits the resemblance between computing the best IE strategy and the problems of MAX-CUT and MAX-DICUT, and builds on the elegant approach of Goemans and Williamson \cite{GW95} and Feige and Goemans \cite{FG95}. Solving the SDP relaxation and using randomized rounding, we obtain a $0.9032$ (resp. $0.9064$) approximation  for the best IE strategy with a pricing probability of $0.586$ for undirected networks (resp. of $2/3$ for directed networks).
Combining these results with the bounds on the fraction of the maximum revenue extracted by the best IE strategy, we significantly improve on the best known approximation ratio for revenue maximization to $0.8229$ for undirected networks and $0.5011$ for directed networks (from $2/3$ and $1/3$, respectively, in \cite{Hart}).

\smallskip\noindent{\bf Other Related Work.}
Our work lies in the area of pricing and revenue maximization in the presence of positive externalities, and more generally, in the wide area of social contagion.
%
%
In this framework, Domingos and Richardson \cite{DR01} studied viral marketing and investigated how a small group of early adopters can be selected, so that the spread of a product is maximized. Subsequently, Kempe, Kleinberg and Tardos \cite{KKT03} considered this question from an algorithmic viewpoint, under the problem of influence maximization, which has received considerable attention since then.

Hartline et al. \cite{Hart} were the first to consider social influence in the framework of revenue maximization. Since then, relevant research has focused either on posted price strategies, where there is no price discrimination, or on game theoretic considerations, where the buyers act strategically according to their perceived value of the product. To the best of our knowledge, our work is the first that considers the approximability of revenue maximization and of computing the best IE strategy, which were the central problems studied in \cite{Hart}. 

Regarding posted price strategies for revenue maximization, Arthur et al. \cite{Arthur} considered a model where the seller cannot approach potential buyers. Instead, only recommendations about the product cascade through the network from an initial seed of early adopters. They gave an Influence-and-Exploit-based constant-factor approximation algorithm for the maximum revenue in this setting.
Akhlaghpour et al. \cite{Iterative} considered iterative posted price strategies where all interested buyers can buy the product at the same price at a given time. They studied the revenue maximization problem under two different repricing models, both allowing for at most $k$ prices. If frequent repricing is allowed, they proved that revenue maximization is $\NP$-hard to approximate, and identified a special case of the problem that can be approximated within reasonable factors. If repricing can be performed only at a limited rate, they presented an FPTAS for revenue maximization.
Anari et al. \cite{Anari} considered a posted price setting where the product exhibits historical externalities. Given a fixed price trajectory, to which the seller commits himself, the buyers decide when to buy the product. In this setting, Anari et al. studied existence and uniqueness of equilibria. They also presented an FPTAS for some special cases of the problem of computing a price trajectory that maximizes the seller's revenue.

In a complementary direction, Chen et al. \cite{Chen} investigated the Nash equilibria and the Bayesian-Nash equilibria when each buyer's value of the product depends on the set of buyers who own the product. They focused on two classes of equilibria, pessimistic and optimistic ones, and showed how to compute these equilibria and how to find revenue-maximizing prices.
Candogan et al. \cite{Bimpikis} investigated a scenario where a monopolist sells a divisible good to buyers under positive externalities. They considered a two-stage game where the seller first sets an individual price for each buyer, and then the buyers decide on their consumption level, according to a utility function of the consumption levels of their neighbors. They proved that the optimal price for each buyer is proportional to his Bonacich centrality, and that if the buyers are partitioned into two pricing classes (which is conceptually similar to Influence-and-Exploit), the problem is reducible to MAX-CUT.

\section{Model and Preliminaries}
\label{s:prelim}

\noindent{\bf The Influence Model.}
The social network is a (possibly directed) weighted network $G(V, E, w)$ on the set $V$ of potential buyers. There is a positive weight $w_{ij}$ associated with each edge $(i, j) \in E$ (we assume that $w_{ij} = 0$ if $(i, j) \not\in E$). A social network is undirected (or symmetric) if $w_{ij} = w_{ji}$ for all $i, j \in V$, and directed (or asymmetric) otherwise. There may exist a non-negative weight $w_{ii}$ associated with each buyer $i$\,%
\footnote{For simplicity, we ignore $w_{ii}$'s for directed social networks. This is without loss of generality, since we can replace each $w_{ii}$ by an edge $(i', i)$ of weight $w_{ii}$ from a new node $i'$ with a single outgoing edge $(i', i)$ and no incoming edges.}.
Every buyer has a value $v_{i}: 2^{N_{i}} \mapsto \reals_{+}$ of the product, which depends on $w_{ii}$ and on the set $S \subseteq N_i$ of his neighbors who already own the product, where $N_i = \{ j \in V \setminus \{i\}: (j, i) \in E \}$. We assume that the exact values of $v_{i}$ are unknown to the seller, and that for each buyer $i$ and each set $S \subseteq N_i$, the seller only knows the probability distribution $F_{i,S}(x) = \Prob[v_{i}(S) < x]$ that buyer $i$ rejects an offer of price $x$ for the product. 

In this work, we focus on the \emph{Uniform Additive Model} \cite[Section~2.1]{Hart}, which can be regarded as an extension of the Linear Threshold Model of social influence introduced in \cite{KKT03}. In the Uniform Additive Model, the values $v_{i}(S)$ are drawn from the uniform distribution in $[0, M_{i,S}]$, where $M_{i,S} = \sum_{j \in S \cup \{i\}} w_{ji}$ is the total influence perceived by $i$, given the set $S$ of his neighbors who own the product. Then, the probability that $i$ rejects an offer of price $x$ is $F_{i,S}(x) = x / M_{i,S}$.

\smallskip\noindent{\bf Myopic Pricing.}
The \emph{myopic price} disregards any externalities imposed by $i$ on his neighbors, and simply maximizes the expected revenue extracted from buyer $i$, given that $S$ is the current set of $i$'s neighbors who own the product. For the Uniform Additive Model, the myopic price is $M_{i,S}/2$, the probability that buyer $i$ accepts it is $1/2$, and the expected revenue extracted from him with the myopic price is $M_{i,S} / 4$, which is the maximum revenue one can extract from buyer $i$ alone.

\smallskip\noindent{\bf Marketing Strategies and Revenue Maximization.}
We can usually extract more revenue from $G$ by employing a marketing strategy that exploits the positive influence between the buyers. A \emph{marketing strategy} $(\vec{\pi}, \vec{x})$ consists of a permutation $\vec{\pi}$ of the buyers and a pricing vector $\vec{x} = (x_1, \ldots, x_n)$, where $\vec{\pi}$ determines the order in which the buyers are approached and $\vec{x}$ the prices offered to them.

We observe that for any buyer $i$ and any probability $p$ that $i$ accepts an offer, there is an (essentially unique) price $x_p$ such that $i$ accepts an offer of $x_p$ with probability $p$. For the Uniform Additive Model, $x_p = (1-p) M_{i,S}$ and the expected revenue extracted from buyer $i$ with such an offer is $p (1-p) M_{i,S}$.
Throughout this paper, we equivalently regard marketing strategies as consisting of a permutation $\vec{\pi}$ of the buyers and a vector $\vec{p} = (p_1, \ldots, p_n)$ of pricing probabilities.
We note that if $p_i = 1$, $i$ gets the product for free, while if $p_i = 1/2$, the price offered to $i$ is (the myopic price of) $M_{i,S}/2$. We assume that $p_i \in [1/2, 1]$, since any expected revenue in $[0, M_{i,S}/4]$ can be achieved with such pricing probabilities. Then, the expected revenue of a marketing strategy $(\vec{\pi}, \vec{p})$ is:
\begin{equation}\label{eq:rev}
   R(\vec{\pi} ,\vec{p}) = \sum_{i \in V} p_i (1 - p_i)
           \left(w_{ii} + \sum_{j: \pi_j < \pi_i} p_{j}w_{ji}\right)
\end{equation}
The problem of \emph{revenue maximization} under the Uniform Additive Model is to find a marketing strategy $(\vec{\pi}^\ast, \vec{p}^\ast)$ that extracts a maximum revenue of $R(\vec{\pi}^\ast, \vec{p}^\ast)$ from a given social network $G(V, E, w)$.

\smallskip\noindent{\bf Bounds on the Maximum Revenue.}
Let $N = \sum_{i \in V} w_{ii}$ and $W = \sum_{i < j} w_{ij}$, if the social network $G$ is undirected, and $W = \sum_{(i, j) \in E} w_{ij}$, if $G$ is directed. Then an upper bound on the maximum revenue of $G$ is $R^\ast = (W+N)/4$, and follows by summing up the myopic revenue over all edges of $G$ \cite[Fact~1]{Hart}. For a lower bound on the maximum revenue, if $G$ is undirected (resp. directed), approaching the buyers in any order (resp. in a random order) and offering them the myopic price yields a revenue of $(W+2N)/8$ (resp. $(W+4N)/16$). Thus, myopic pricing achieves an approximation ratio of $0.5$ for undirected networks and of $0.25$ for directed networks.

\smallskip\noindent{\bf Ordering and $\NP$-Hardness.}
Revenue maximization exhibits a dual nature involving optimizing both the pricing probabilities and the sequence of offers.
For directed networks, finding a good ordering $\vec{\pi}$ of the buyers bears a resemblance to the Maximum Acyclic Subgraph problem, where given a directed network $G(V, E, w)$, we seek for an acyclic subgraph of maximum total edge weight. In fact, any permutation $\vec{\pi}$ of $V$ corresponds to an acyclic subgraph of $G$ that includes all edges going forward in $\vec{\pi}$, i.e, all edges $(i, j)$ with $\pi_i < \pi_j$.
\cite[Lemma~3.2]{Hart} shows that given a directed network $G$ and a pricing probability vector $\vec{p}$, computing an optimal ordering of the buyers (for the particular $\vec{p}$) is equivalent to computing a Maximum Acyclic Subgraph of $G$, with each edge $(i, j)$ having a weight of $p_i p_j (1-p_j) w_{ij}$. Consequently, computing an ordering $\vec{\pi}$ that maximizes $R(\vec{\pi}, \vec{p})$ is $\NP$-hard and Unique-Games-hard to approximate within a factor greater than 0.5 \cite{GMR08}.

On the other hand, we show that in the undirected case, if the pricing probabilities are given, we can easily compute the best ordering of the buyers (see also Section~\ref{app:cycle}, in the Appendix, for a simple example about the importance of a good ordering in the undirected case).

\begin{lemma}\label{l:ordering}
Let $G(V, E, w)$ be an undirected social network, and let $\vec{p}$ be any pricing probability vector. Then, approaching the buyers in non-increasing order of their pricing probabilities maximizes the revenue extracted from $G$ under $\vec{p}$.
\end{lemma}

\begin{proof}
We consider an optimal ordering $\vec{\pi}$ (wrt. $\vec{p}$) that minimizes the number of buyers' pairs appearing in increasing order of their pricing probabilities, namely, the number of pairs $i_1, i_2$ with $p_{i_1} < p_{i_2}$ and $\pi_{i_1} < \pi_{i_2}$. If there is such a pair in $\vec{\pi}$, we can find a pair of buyers $i$ and $j$ with $p_i < p_j$ such that $i$ appears just before $j$ in $\vec{\pi}$. Then, switching the positions of $i$ and $j$ in $\vec{\pi}$ changes the expected revenue extracted from $G$ under $\vec{p}$ by $p_i p_j w_{ij} (p_j - p_i) \geq 0$, a contradiction.
 \qed\end{proof}

A consequence of Lemma~\ref{l:ordering} is that \cite[Lemma~3.2]{Hart} does not imply the $\NP$-hardness of revenue maximization for undirected social networks. The following lemma employs a reduction from monotone One-in-Three 3-SAT, and shows that revenue maximization is $\NP$-hard for undirected networks.

\begin{lemma}\label{l:maxrev-hard}
The problem of computing a marketing strategy that extracts the maximum revenue from an undirected social network is $\NP$-hard.
\end{lemma} 

\begin{proof}
In monotone One-in-Three 3-SAT, we are given a set $V$ of $n$ items and $m$ subsets $T_1, \ldots, T_m$ of $V$, with $2 \leq |T_j| \leq 3$ for each $j \in \{ 1, \ldots, m \}$. We ask for a subset $S \subset V$ such that $|S \cut T_j| = 1$ for all $j \in \{ 1, \ldots, m \}$. Monotone One-in-Three 3-SAT is shown $\NP$-complete in \cite{SAT}. In the following, we let $m_2$ (resp. $m_3$) denote the number of $2$-item (resp. $3$-item) sets $T_j$ in an instance $(V, T_1, \ldots, T_m)$ of monotone One-in-Three 3-SAT.

\insertfig{0.6\textwidth}{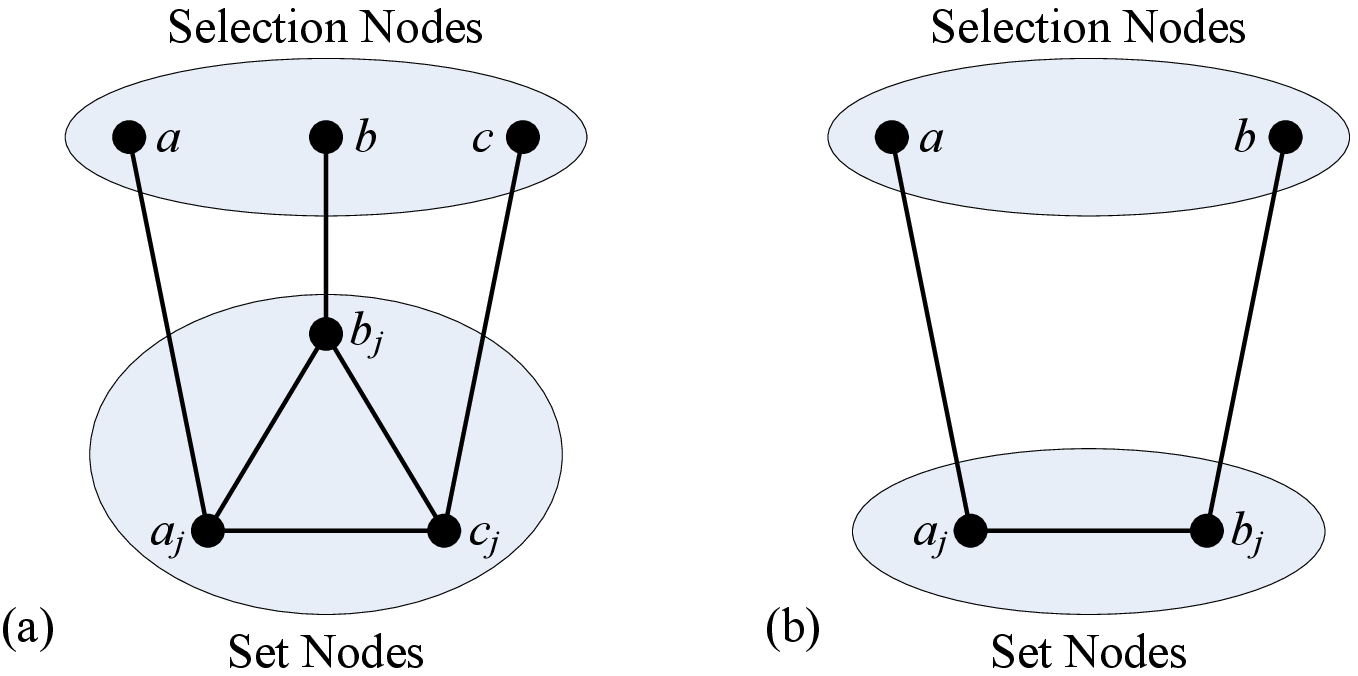}{\label{fig:maxrev-hard}Examples of (a) an extended triangle and (b) a $3$-path, used in the proof of Lemma~\ref{l:maxrev-hard}. We create an extended triangle for each $3$-item set $T_j$ and a $3$-path for each $2$-item set $T_j$. The set nodes are different for each set $T_j$, while the selection nodes are common for all sets.}

Given $(V, T_1, \ldots, T_m)$, we construct an undirected social network $G$. The network $G$ contains a \emph{selection-node} corresponding to each item in $V$. There are no edges between the selection nodes of $G$. For each $3$-item set $T_j = \{a, b, c\}$, we create an \emph{extended triangle} consisting of a triangle on three \emph{set nodes} $a_j$, $b_j$, and $c_j$, and three additional edges that connect $a_j$, $b_j$, $c_j$ to the corresponding selection nodes $a$, $b$, and $c$ (see also Fig.~\ref{fig:maxrev-hard}.a). For each $2$-item set $T_j = \{a, b\}$, we create a \emph{$3$-path} consisting of an edge connecting two set nodes $a_j$ and $b_j$, and two additional edges connecting $a_j$ and $b_j$ to the corresponding selection nodes $a$ and $b$ (see also Fig.~\ref{fig:maxrev-hard}.b). Therefore, $G$ contains $n+2m_2+3m_3$ nodes and $3m_2+6m_3$ edges. The weight of all edges of $G$ is $1$.
We next show that $(V, T_1, \ldots, T_m)$ is a {\small YES}-instance of monotone One-in-Three 3-SAT iff the maximum revenue of $G$ is at least
$\frac{177}{128}\,m_3 + \frac{3}{4}\,m_2$.

By Lemma~\ref{l:ordering}, the revenue extracted from $G$ is maximized if the nodes are approached in non-increasing order of their pricing probabilities. Therefore, we can ignore the ordering of the nodes, and focus on their pricing probabilities.
The important property is that if each extended triangle (Fig.~\ref{fig:maxrev-hard}.a) is considered alone, its maximum revenue is $177/128$, and is obtained when exactly one of the selection nodes $a, b, c$ has a pricing probability of $1/2$ and the other two have a pricing probability of $1$.
More specifically, since the selection nodes $a, b, c$ have degree $1$, the revenue of the extended triangle is maximized when they have a pricing probability of either $1$ or $1/2$.
If all $a, b, c$ have a pricing probability of $1$, the best revenue of the extended triangle is $\approx 1.196435$, and is obtained when one of $a_j$, $b_j$, and $c_j$ has a pricing probability of $\approx 0.7474$, the other has a pricing probability of $\approx 0.5715$, and the third has a pricing probability of $1/2$.
If all $a, b, c$ have a pricing probability of $1/2$, the best revenue of the extended triangle is again $\approx 1.196435$, and is obtained with the same pricing probabilities of $a_j$, $b_j$, and $c_j$.
If two of $a, b, c$ (say $a$ and $b$) have a pricing probability of $1/2$ and $c$ has a pricing probability of $1$, the best revenue of the extended triangle is $\frac{21}{16} = 1.3125$, and is obtained when one of $a_j$ and $b_j$ has a pricing probability of $1$, the other has a pricing probability of $3/4$, and $c_j$ has a pricing probability of $1/2$.
Finally, if two of $a, b, c$ (say $b$ and $c$) have a pricing probability of $1$ and $a$ has a pricing probability of $1/2$, we extract a maximum revenue from the extended triangle, which is $\frac{177}{128} = 1.3828125$ and is obtained when $a_j$ has a pricing probability of $1$, one of $b_j$ and $c_j$ has a pricing probability of $9/16$, and the other has a pricing probability of $1/2$.

Similarly, if each $3$-path (Fig.~\ref{fig:maxrev-hard}.b) is considered alone, its maximum revenue is $3/4$, and is obtained when exactly one of the selection nodes $a, b$ has a pricing probability of $1/2$ and the other has a pricing probability of $1$.
In fact, since the $3$-path is a bipartite graph, Proposition~\ref{pr:bipartite} implies that the maximum revenue, which is $3/4$, is extracted when $a_j$ and $b$ have a pricing probability of $1$ and $b_j$ and $a$ have a pricing probability of $1/2$ (or the other way around).
If both $a$ and $b$ have a pricing probability of $1$, the best revenue of the $3$-path is $41/64$ and is obtained when one of $a_j$ and $b_j$ has a pricing probability of $5/8$, and the other has a pricing probability of $1/2$.
If both $a$ and $b$ have a pricing probability of $1/2$, the best revenue of $3$-path is again $41/64$ and is obtained when one of $a_j$ and $b_j$ has a pricing probability of $1$, and the other has a pricing probability of $5/8$.

If $(V, T_1, \ldots, T_m)$ is a {\small YES}-instance of monotone One-in-Three 3-SAT, we assign a pricing probability of $1/2$ to the selection nodes in $S$ and a pricing probability of $1$ to the selection nodes in $V \setminus S$, where $S$ is a set with exactly one element of each $T_j$. Thus, we have exactly one selection node with pricing probability $1/2$ in each extended triangle and in each $3$-path. Then, we can set the pricing probabilities of the set nodes as above, so that the revenue of each extended triangle is $177/128$ and the revenue of each $3$-path is $3/4$. Thus, the maximum revenue of $G$ is at least
$\frac{177}{128}\,m_3 + \frac{3}{4}\,m_2$.

For the converse, we recall that the edges of $G$ can be partitioned into $m_3$ extended triangles and $m_2$ $3$-paths. Consequently, if the maximum revenue of $G$ is at least $\frac{177}{128}\,m_3 + \frac{3}{4}\,m_2$, each extended triangle contributes exactly $177/128$ and each $3$-path contributes exactly $3/4$ to the revenue of $G$. Thus, by the analysis on their revenue above, each extended triangle and each $3$-path includes exactly one selection node with a pricing probability of $1/2$. Therefore, if we let $S$ consist of the selection nodes with pricing probability $1/2$, we have that $|S \cut T_j| = 1$ for all $j \in \{ 1, \ldots, m \}$. 
 \qed\end{proof}

\section{Influence-and-Exploit Strategies}
\label{s:ie}

An \emph{Influence-and-Exploit} (IE) strategy $\IE(A, p)$ consists of a set of buyers $A$ receiving the product for free and a pricing probability $p$ offered
%
%
to the remaining buyers in $V \setminus A$, who are approached in a random order. We slightly abuse the notation and let $\IE(q, p)$ denote an IE strategy where each buyer is selected in $A$ independently with probability $q$. $\IE(A, p)$ extracts an expected (wrt the random ordering of the exploit set) revenue of:
\begin{equation}\label{eq:ie-revenue}
 \RIE(A, p) = p (1-p) \sum_{i \in V \setminus A} \left(w_{ii}+
          \sum_{j \in A} w_{ji} +
          \sum_{j \in V \setminus A,\ j \neq i} \frac{ p\,w_{ji}}{2} \right)
\end{equation}
Specifically, $\IE(A, p)$ extracts a revenue of $p(1-p)w_{ji}$ from each edge $(j, i)$ with buyer $j$ in the influence set $A$ and buyer $i$ in the exploit set $V \setminus A$. Moreover, $\IE(A, p)$ extracts a revenue of $p^2(1-p)w_{ji}$ from each edge $(j, i)$ with both $j, i$ in the exploit set, if $j$ appears before $i$ in the random order of $V \setminus A$, which happens with probability $1/2$.

The problem of finding the best IE strategy is to compute a subset of buyers $A^\ast$ and a pricing probability $p^\ast$ that extract a maximum revenue of $\RIE(A^\ast, p^\ast)$ from a given social network $G(V, E, w)$.
The following lemma employs a reduction from monotone One-in-Three 3-SAT, and shows that computing the best IE strategy is $\NP$-hard.

\begin{lemma}\label{l:ie-hard}
Let $p \in [1/2, 1)$ be any fixed pricing probability. The problem of finding the best IE strategy with pricing probability $p$ is $\NP$-hard, even for undirected social networks.
\end{lemma}

\begin{proof}
We recall that in monotone One-in-Three 3-SAT, we are given a set $V$ of $n$ items  and $m$ subsets $T_1, \ldots, T_m$ of $V$, with $2 \leq |T_j| \leq 3$ for each $j \in \{ 1, \ldots, m \}$. We ask for a subset $S \subset V$ such that $|S \cut T_j| = 1$ for all $j \in \{ 1, \ldots, m \}$.

Given $(V, T_1, \ldots, T_m)$, we construct an undirected social network $G$ on $V$. For each $3$-item set $T_j = \{a, b, c\}$, we create a \emph{set-triangle} on nodes $a$, $b$, and $c$ with $3$ edges of weight $1$. For each $2$-item set $T_j = \{a, b\}$, we add a \emph{set-edge} $\{ a, b \}$ of weight $2+p$, where $p$ is the pricing probability. To avoid multiple appearances of the same edge, we let the weight of each edge be the total weight of its appearances. Namely, if an edge $e$ appears in $k_3$ set-triangles and in $k_2$ set-edges, $e$'s weight is $k_3 + (2+p)k_2$. We observe that for any $p \in [1/2, 1)$, the maximum revenue extracted from any set-triangle and any set-edge is $p(1-p)(2+p)$, by giving the product for free to exactly one of the nodes of the set-triangle (resp. the set-edge).

We next show that $(V, T_1, \ldots, T_m)$ is a {\small YES}-instance of monotone One-in-Three 3-SAT iff there is an influence set $A$ in $G$ such that $\RIE(A, p) \geq m p (1-p) (2+p)$.
If $(V, T_1, \ldots, T_m)$ is a {\small YES}-instance of monotone One-in-Three 3-SAT, we let the influence set $A = S$, where $S$ is a set with exactly one element of each $T_j$. Then, we extract an expected revenue of $p(1-p)(2+p)$ from each set-triangle and each set-edge in $G$, which yields an expected revenue of $mp(1-p)(2+p)$ in total.
For the converse, if there is an influence set $A$ in $G$ such that $\RIE(A, p) \geq m p (1-p) (2+p)$, we let $S = A$. Since $\RIE(A, p) \geq m p (1-p) (2+p)$, and since the edges of $G$ can be partitioned into $m$ set-triangles and set-edges, each with a maximum revenue of at most $p(1-p)(2+p)$, each set-triangle and each set-edge contributes exactly $p (1-p) (2+p)$ to $\RIE(A, p)$. Therefore, for all set-triangles and all set-edges, there is exactly one node in $A$. Thus, we have that $|S \cut T_j| = 1$ for all $j \in \{ 1, \ldots, m \}$.
 \qed\end{proof}

Interestingly, even very simple IE strategies extract a significant fraction of the maximum revenue. For example, for undirected social networks, $\RIE(\emptyset, 2/3) = (4W+6N)/27$, and thus $\IE(\emptyset, 2/3)$ achieves an approximation ratio of $\frac{16}{27} \approx 0.592$. For directed networks, $\RIE(\emptyset, 2/3) = (2W+6N)/27$, and thus $\IE(\emptyset, 2/3)$ achieves an approximation ratio of $\frac{8}{27} \approx 0.296$. In the following, we show that carefully selected IE strategies manage to extract a larger fraction of the maximum revenue.

\smallskip\noindent{\bf Exploiting Large Cuts.}
A natural idea is to exploit the apparent connection between a large cut in the social network and a good IE strategy. For example, in the undirected case, an IE strategy $\IE(q, p)$ is conceptually similar to the randomized $0.5$-approximation algorithm for MAX-CUT, which puts each node in set $A$ with probability $1/2$. However, in addition to a revenue of $p(1-p)w_{ij}$ from each edge $\{ i, j\}$ in the cut $(A, V \setminus A)$, $\IE(q, p)$ extracts a revenue of $p^2(1-p)w_{ij}$ from each edge $\{ i, j \}$ between nodes in the exploit set $V \setminus A$. Thus, to optimize the performance of $\IE(q, p)$, we carefully adjust the probabilities $q$ and $p$ so that $\IE(q, p)$ balances between the two sources of revenue. Hence, we obtain the following:

\begin{proposition}\label{pr:simple-ie}
Let $G(V, E, w)$ be an undirected social network, and let $q = \max\{ 1-\frac{\sqrt{2}(2+\lambda)}{4}, 0\}$, where $\lambda = N/W$. Then, $\IE(q, 2-\sqrt{2})$ approximates the maximum revenue extracted from $G$ within a factor of at least $2\sqrt{2}(2-\sqrt{2})(\sqrt{2}-1) \approx 0.686$.
\end{proposition}

\begin{proof}
The proof extends the proof of \cite[Theorem~3.1]{Hart}. We start with calculating the expected (wrt to the random choice of the influence set) revenue of $\IE(q, p)$.
The expected revenue of $\IE(q, p)$ from each loop $\{ i, i \}$ is $(1-q)p(1-p)w_{ii}$. In particular, a revenue of $p(1-p)w_{ii}$ is extracted from $\{i, i\}$  if buyer $i$ is included in the exploit set, which happens with probability $1-q$.
The expected revenue of $\IE(q, p)$ from each edge $\{ i, j \}$, $i < j$, is $(2q(1-q)p(1-p)+(1-q)^2p^2(1-p))w_{ij}$. More specifically, if one of $i$, $j$ is included in the influence set and the other is included in the exploit set, which happens with probability $2q(1-q)$, a revenue of $p(1-p)w_{ij}$ is extracted from edge $\{i, j\}$.
Otherwise, if both $i$ and $j$ are included in the exploit set, which happens with probability $(1-q)^2$, a revenue of $p^2(1-p)w_{ij}$ is extracted from edge $\{i, j\}$ (note that since $\{i, j\}$ is an undirected edge, the order in which $i$ and $j$ are considered in the exploit set is insignificant).
By linearity of expectation, the expected revenue of $\IE(q, p)$ is:
\[ \RIE(q, p) = (1-q) p (1-p) \sum_{i \in V} w_{ii} +
                  (1-q) p (1-p) \sum_{i < j} (2q + p(1-q)) w_{ij} \]
Using that $N = \sum_{i \in V} w_{ii}$ and $W = \sum_{i < j} w_{ij}$, and setting $N = \lambda W$, we obtain that:
\[ \RIE(q, p) = (1-q) p (1-p) (\lambda + 2q + p(1-q)) W \]
Differentiating with respect to $q$, we obtain that the optimal value of $q$ is
\[ q^\ast = \max\left\{ \frac{1-p-\lambda/2}{2-p}, 0 \right\}\]
We recall that $R^\ast = (1+\lambda) W / 4$ is an upper bound on the maximum revenue of $G$. Therefore, the approximation ratio of $\IE(q, p)$ is:
\begin{equation}\label{eq:ie_optimized}
 \frac{4 (1-q) p (1-p)(\lambda + 2q + p(1-q))}{1+\lambda}
\end{equation}
Using $p = 1/2$ and $q = \max\left\{\frac{1-\lambda}{3}, 0\right\}$ in (\ref{eq:ie_optimized}), we obtain the IE strategy of \cite[Theorem~3.1]{Hart}, whose approximation ratio is at least $2/3$, attained at $\lambda = 0$. Assuming small values of $\lambda$, so that $q^\ast > 0$, and differentiating with respect to $p$, we obtain that the best value of $p$ for $\IE(q^\ast, p)$  is $p^\ast = 2 - \sqrt{2}$. Using $p = 2 - \sqrt{2}$ and $q = \max\left\{ 1-\frac{\sqrt{2}(2+\lambda)}{4}, 0\right\}$, we obtain an IE strategy with an approximation ratio of at least $2\sqrt{2}(2-\sqrt{2})(\sqrt{2}-1) \approx 0.686$, attained at $\lambda = 0$.
 \qed\end{proof}

\begin{proposition}\label{pr:simple-ie-directed}
Let $G(V, E, w)$ be a directed social network. Then, $\IE\!\left(1-\frac{\sqrt{2}}{2}, 2-\sqrt{2}\right)$ approximates the maximum revenue of $G$ within a factor of $\sqrt{2}(2-\sqrt{2})(\sqrt{2}-1) \approx 0.343$.
\end{proposition}

\begin{proof}
The proof is similar to the proof of Proposition~\ref{pr:simple-ie}. We recall that for the directed case, we can ignore loops $(i, i)$. Since the social network $G$ is directed, the expected (wrt to the random choice of the influence set and the random order of the exploit set) revenue of $\IE(q, p)$ is:
\begin{eqnarray*}
 \RIE(q, p) & = &
   (1-q) p (1-p) \sum_{(i, j) \in E} (q + p(1-q)/2) w_{ij} \\
 & = & (1-q) p (1-p) (q + p(1-q)/2) W
\end{eqnarray*}
More specifically, if $i$ is included in the influence set and $j$ is included in the exploit set, which happens with probability $q(1-q)$, a revenue of $p(1-p)w_{ij}$ is extracted from each edge $(i, j)$. Furthermore, if both $i$ and $j$ are included in the exploit set $V \setminus A$ and $i$ appears before $j$ in the random order of $V \setminus A$, which happens with probability $(1-q)^2/2$, a revenue of $p^2(1-p)w_{ij}$ is extracted from edge $(i, j)$.
%
%

Using the upper bound of $W / 4$ on the maximum revenue of $G$, we have that the approximation ratio of $\IE(q, p)$ is at least
$4 (1-q) p (1-p)(q + p(1-q)/2)$.
Setting $q = 1/3$ and $p = 1/2$, we obtain the IE strategy of \cite[Theorem~3.1]{Hart}, whose approximation ratio for directed networks is $1/3$. Using $q = 1-\frac{\sqrt{2}}{2}$ and $p = 2 - \sqrt{2}$, we obtain an IE strategy with an approximation ratio of $\sqrt{2}(2-\sqrt{2})(\sqrt{2}-1) \approx 0.343$.
 \qed\end{proof}

\begin{proposition}[Optimality of IE for Bipartite Networks]\label{pr:bipartite}
Let $G(V, E, w)$ be an undirected bipartite social network with $w_{ii} = 0$ for all buyers $i$, and let $(A, V \setminus A)$ be any partition of $V$ into independent sets. Then, $\IE(A, 1/2)$ extracts the maximum revenue of $G$.
\end{proposition}

\begin{proof}
Since all edges of $G$ are between buyers in the influence set $A$ and buyers in the exploit set $V \setminus A$, $\IE(A, 1/2)$ extracts the myopic revenue of $w_{ij}/4$ from any edge $\{ i, j \} \in E$. Therefore, $\IE(A, 1/2)$ is an optimal strategy.
 \qed\end{proof}

\section{On the Efficiency of Influence-and-Exploit Strategies}

Next, we show that the best IE strategy, which is $\NP$-hard to compute, manages to extract a significant fraction of the maximum revenue.

\begin{theorem}\label{th:bestIE-undir}
For any undirected social network, there is an IE strategy with pricing probability $0.586$ whose revenue is at least $0.9111$ times the maximum revenue.
\end{theorem}

\begin{proof}
We consider an arbitrary undirected social network $G(V, E, w)$, start from an arbitrary pricing probability vector $\vec{p}$, and obtain an IE strategy $\IE(A, \hat{p})$ by applying randomized rounding to $\vec{p}$. We show that for $\hat{p} = 0.586$, the expected (wrt the randomized rounding choices) revenue of $\IE(A, \hat{p})$ is at least $0.9111$ times the revenue extracted from $G$ by the best ordering for $\vec{p}$ (recall that by Lemma~\ref{l:ordering}, the best ordering is to approach the buyers in non-increasing order of their pricing probabilities).

Without loss of generality, we assume that $p_1 \geq p_2 \geq \cdots \geq p_n$, and let $\vec{\pi}$ be the identity permutation. Then,
\(
 R(\vec{\pi}, \vec{p}) =
 \sum_{i \in V} p_i(1-p_i) w_{ii} + \sum_{i < j} p_i p_j (1-p_j) w_{ij} \).

For the IE strategy, we assign each buyer $i$ to the influence set $A$ independently with probability $I(p_i) = \alpha(p_i - 0.5)$, for some appropriate $\alpha \in [0, 2]$, and to the exploit set with probability $E(p_i) = 1 - I(p_i)$. By linearity of expectation, the expected revenue of $\IE(A, \hat{p})$ is:
\[
 \RIE(A, \hat{p}) =  \sum_{i \in V} \hat{p}(1 - \hat{p})E(p_i)w_{ii} +
 \sum_{i < j} \hat{p} (1 - \hat{p}) ( I(p_i) E(p_j) + E(p_i) I(p_j) +
   \hat{p}\,E(p_i) E(p_j)) w_{ij}
\]
Specifically, $\IE(A, \hat{p})$ extracts a revenue of $\hat{p}(1-\hat{p}) w_{ii}$ from edge loop $\{i, i\}$, if $i$ is included in the exploit set. Moreover, $\IE(A, \hat{p})$ extracts a revenue of $\hat{p}(1-\hat{p}) w_{ij}$ from each edge $\{i, j\}$, $i < j$, if one of $i$, $j$ is included in the influence set $A$ and the other is not, and a revenue of $\hat{p}^2(1-\hat{p})w_{ij}$ if both $i$ and $j$ are included in the exploit set $V \setminus A$ (note that the order in which $i$ and $j$ are considered is insignificant).
%

The approximation ratio is derived as the minimum ratio between any pair of terms in $R(\vec{\pi}, \vec{p})$ and $\RIE(A, \hat{p})$ corresponding to the same loop $\{ i, i \}$ or to the same edge $\{ i, j \}$. For a weaker bound, we observe that for $\alpha = 1.43$ and $\hat{p} = 0.586$, both
\begin{equation}\label{eq:bestIE-undir-rat}
 \min_{0.5 \leq x \leq 1}\frac{\hat{p}\,(1 - \hat{p})\,E(x)}{x\,(1-x)}
\ \ \ \mbox{and}\ \ \
 \min_{0.5 \leq y \leq x \leq 1}
 \frac{\hat{p}\,(1 - \hat{p})(I(x)\,E(y) + E(x)\,I(y) + \hat{p}\,E(x)\, E(y))}{x\,y\,(1-y)}
\end{equation}
are at least $0.8024$. More precisely, the former quantity is minimized for $x \approx 0.7104$, for which it becomes $\approx 0.8244$. For any fixed value of $y \in [0.5, 1.0]$, the latter quantity is minimized for $x = 1.0$. The minimum value is $0.8024$ for $x = 1.0$ and $y \approx 0.629$.

For the stronger bound of $0.9111$, we let $\hat{p} = 0.586$, and for each buyer $i$, let the rounding parameter $\alpha(p_i)$ be chosen according to the following piecewise linear function of $p_i$\,:
\[ \alpha(p_i) = \left\{\begin{array}{ll}
5.0\,(p_i - 0.5) & \mbox{if $0.5 \leq p_i \leq 0.7$}\\
1.0 + 3.3\,(p_i - 0.7) & \mbox{if $0.7 < p_i \leq 0.8$}\\
1.33 + 3.0\,(p_i - 0.8) & \mbox{if $0.8 < p_i \leq 0.9$}\\
1.63 + 3.7\,(p_i - 0.9)\ \ \ & \mbox{if $0.9 < p_i \leq 1.0$}\\
\end{array}\right.\]

The quantity on the left of (\ref{eq:bestIE-undir-rat}) is minimized for $x = 0.8$, for which it becomes $\approx 0.9112$. For any fixed $x \in [0.5, 0.949]$, the quantity on the right of (\ref{eq:bestIE-undir-rat}) is minimized for $y = 0.5$. The minimum value is $0.9111$ for $x \approx 0.7924$ and $y = 0.5$. For any $x \in (0.949, 0.983]$, the latter quantity is minimized for $y = 0.7$. The minimum value, over all $x \in (0.949, 0.983]$, is $\approx 0.93$ at $x = 0.983$ and $y = 0.7$. For any fixed $x \in (0.983, 1.0]$, the quantity on the right of (\ref{eq:bestIE-undir-rat}) is minimized for some $y \in [0.7, 0.8]$. Moreover, for all $y \in [0.7, 0.8]$, this quantity is minimized for $x = 1.0$. The minimum value is $\approx 0.9112$ at $x = 1.0$ and $y \approx 0.8$.
 \qed\end{proof}

\begin{theorem}\label{th:bestIE-dir}
For any directed social network, there is an IE strategy with pricing probability $2/3$ whose expected revenue is at least $0.55289$ times the maximum revenue.
\end{theorem}

\begin{proof}
As before, we consider an arbitrary directed social network $G(V, E, w)$, start from an arbitrary pricing probability vector $\vec{p}$, and obtain an IE strategy $\IE(A, \hat{p})$ by applying randomized rounding to $\vec{p}$. We show that for $\hat{p} = 2/3$, the expected (wrt the randomized rounding choices) revenue of $\IE(A, \hat{p})$ is at least $0.55289$ times the revenue extracted from $G$ under the best ordering for $\vec{p}$ (which ordering is Unique-Games-hard to approximate within a factor less than $0.5$!).

We recall that in the directed case, we can, without loss of generality, ignore loops $(i, i)$. Let $\vec{\pi}$ be the best ordering $\vec{\pi}$ for $\vec{p}$. Then, the maximum revenue extracted from $G$ with pricing probabilities $\vec{p}$ is
\(
 R(\vec{\pi}, \vec{p}) \leq \sum_{(i, j) \in E} p_i p_j (1-p_j) w_{ij} \).

As in the proof of Theorem~\ref{th:bestIE-undir}, we assign each buyer $i$ to the influence set $A$ independently with probability $I(p_i) = \alpha(p_i - 0.5)$, for some $\alpha \in [0, 2]$, and to the exploit set with probability $E(p_i) = 1 - I(p_i)$. By linearity of expectation, the expected (wrt the randomized rounding choices) revenue extracted by $\IE(A, \hat{p})$ is:
\[
 \RIE(A, \hat{p}) =
 \sum_{(i, j) \in E} \hat{p} (1 - \hat{p}) ( I(p_i) E(p_j) +
   0.5\,\hat{p}\,E(p_i) E(p_j)) w_{ij}
\]
Specifically, $\IE(A, \hat{p})$ extracts a revenue of $\hat{p}(1-\hat{p}) w_{ij}$ from each edge $(i, j)$, if $i$ is included in the influence set and $j$ is included in the exploit set, and a revenue of $\hat{p}^2(1-\hat{p})w_{ij}$ if both $i$ and $j$ are included in the exploit set $V \setminus A$ and $i$ appears before $j$ in the random order of $V \setminus A$.
%

The approximation ratio is derived as the minimum ratio between any pair of terms in $R(\vec{\pi}, \vec{p})$ and $\RIE(A, \hat{p})$ corresponding to the same edge $(i, j)$. Thus, we select $\hat{p}$ and $\alpha$ so that the following quantity is maximized:
\[
 \min_{0.5 \leq x, y \leq 1}
 \frac{\hat{p}\,(1 - \hat{p})(I(x)\,E(y) + 0.5\,\hat{p}\,E(x)\,E(y))}{x\,y\,(1-y)}
\]
We observe that for $\hat{p} = 2/3$ and $\alpha = 1.0$, this quantity is simplified to $\min_{y \in [0.5, 1]} \frac{2(3-2y)}{27y(1-y)}$. The minimum value is $\approx 0.55289$ at $y = \frac{3-\sqrt{3}}{2}$.
 \qed\end{proof}

Similarly, we can show that there is an IE strategy with pricing probability $1/2$ whose revenue is at least $0.8857$ (resp. $0.4594$) times the maximum revenue for undirected (resp. directed) networks.

\subsection{On the Approximability of the Maximum Revenue for Directed Networks}

The results of \cite[Lemma~3.2]{Hart} and \cite{GMR08} suggest that given a pricing probability vector $\vec{p}$, it is Unique-Games-hard to compute a vertex ordering $\vec{\pi}$ of a directed network $G$ for which the revenue of $(\vec{\pi}, \vec{p})$ is at least $0.5$ times the maximum revenue of $G$ under $\vec{p}$.
An interesting consequence of Theorem~\ref{th:bestIE-dir} is that this inapproximability bound of $0.5$ does not apply to revenue maximization in the Uniform Additive Model. In particular, given a pricing probability vector $\vec{p}$, Theorem~\ref{th:bestIE-dir} constructs, in linear time, an IE strategy with an expected revenue of at least $0.55289$ times the maximum revenue of $G$ under $\vec{p}$. This does not contradict the results of \cite{Hart,GMR08}, because the pricing probabilities of the IE strategy are different from $\vec{p}$.
Moreover, in the Uniform Additive Model, different acyclic (sub)graphs (equivalently, different vertex orderings) allow for a different fraction of their edge weight to be translated into revenue (for an example, see Section~\ref{app:approximability}, in the Appendix), while in the reduction of \cite[Lemma~3.2]{Hart}, the weight of each edge in an acyclic subgraph is equal to its revenue.
Thus, although the IE strategy of Theorem~\ref{th:bestIE-dir} is $0.55289$-approximate with respect to the maximum revenue of $G$ under $\vec{p}$, its vertex ordering combined with $\vec{p}$ may generate a revenue of less than $0.5$ times the maximum revenue of $G$ under $\vec{p}$.
In fact, based on Theorem~\ref{th:bestIE-dir}, we obtain, in Section~\ref{s:sdp}, a polynomial-time algorithm that approximates the maximum revenue of a directed network $G$ within a factor of $0.5011$.

The following propositions establish a pair of inapproximabity results for revenue maximization in the Uniform Additive Model.


\begin{proposition}\label{pr:approx-t1}
Assuming the Unique Games conjecture, it is $\NP$-hard to compute an IE strategy with pricing probability $2/3$ that approximates within a factor greater than $3/4$ the maximum revenue of a directed social network in the Uniform Additive Model.
\end{proposition}

\begin{proof}
Let $G(V, E, w)$ be a directed social network, and
%
%
let $\vec{\pi}^\ast$ be a vertex ordering corresponding to an acyclic subgraph of $G$ with a maximum edge weight of $W^\ast$. Then, approaching the buyers according to $\vec{\pi}^\ast$ and offering a pricing probability of $2/3$ to each of them, we extract a revenue of $4W^\ast/27$. Therefore, the maximum revenue of $G$ is at least $4W^\ast/27$.

Now, we assume an influence set $A$ so that $\IE(A, 2/3)$ approximates the maximum revenue of $G$ within a factor of $r$. Thus, $\RIE(A, 2/3) \geq 4 r W^\ast / 27$. Let $\vec{\pi}$ be the order in which $\IE(A, 2/3)$ approaches the buyers, and let $(i, j)$ be any edge with $\pi_i < \pi_j$, namely, any edge from which $\IE(A, 2/3)$ extracts some revenue. Since the revenue extracted from each such edge $(i, j)$ is at most $2 w_{ij} / 9$, the edge weight of the acyclic subgraph defined by $\vec{\pi}$ is at least
\( \tfrac{9}{2} \RIE(A, 2/3) \geq \tfrac{2 r}{3} W^\ast \).

Hence, given an $r$-approximate $\IE(A, 2/3)$, we can approximate $W^\ast$ within a ratio of $2r/3$. The proposition follows from \cite[Theorem~1.1]{GMR08}, which assumes the Unique Games conjecture and shows that it is $\NP$-hard to approximate $W^\ast$ within a ratio greater than $1/2$.
 \qed\end{proof}

\begin{proposition}\label{pr:approx-t2}
Assuming the Unique Games conjecture, it is $\NP$-hard to approximate within a factor greater than $27/32$ the maximum revenue of a directed social network in the Uniform Additive Model.
\end{proposition}

\begin{proof}
The proof is similar to the proof of Proposition~\ref{pr:approx-t1}. Let $G(V, E, w)$ be a directed social network, and let $\vec{\pi}^\ast$ be a vertex ordering corresponding to an acyclic subgraph of $G$ with a maximum edge weight of $W^\ast$. Using $\vec{\pi}^\ast$ and a pricing probability of $2/3$ for all buyers, we obtain that the maximum revenue of $G$ is at least $4W^\ast/27$.

We assume a marketing strategy $(\vec{\pi}, \vec{p})$ that approximates the maximum revenue of $G$ within a factor of $r$. Thus, $R(\vec{\pi}, \vec{p}) \geq 4 r W^\ast / 27$. Let $(i, j)$ be any edge with $\pi_i < \pi_j$, namely, any edge from which $(\vec{\pi}, \vec{p})$ extracts some revenue. Since the revenue extracted from each such edge $(i, j)$ is at most $w_{ij} / 4$, the edge weight of the acyclic subgraph defined by $\vec{\pi}$ is at least
\( 4 R(\vec{\pi}, \vec{p}) \geq \tfrac{16 r}{27} W^\ast \)

Thus, given an $r$-approximate marketing strategy $(\vec{\pi}, \vec{p})$, we can approximate $W^\ast$ within a ratio of $16 r / 27$. Now, the proposition follows from \cite[Theorem~1.1]{GMR08}.
 \qed\end{proof}

\section{Generalized Influence-and-Exploit}
\label{s:ie-gen}

Building on the idea of generating revenue from large cuts between different pricing classes, we obtain a class of generalized IE strategies, which employ a refined partition of buyers in more than two pricing classes. We first analyze the efficiency of generalized IE strategies for undirected networks, and then translate our results to the directed case. The analysis generalizes the proof of Proposition~\ref{pr:simple-ie}.

A generalized IE strategy consists of $K$ pricing classes, for some appropriately large integer $K \geq 2$. Each class $k$, $k = 1, \ldots, K$, is associated with a pricing probability of $p_k = 1 - \frac{k-1}{2(K - 1)}$. Each buyer is assigned to the pricing class $k$ independently with probability $q_k$, where $\sum_{k=1}^K q_k = 1$, and is offered a pricing probability of $p_k$. The buyers are considered in non-increasing order of their pricing probabilities, i.e., the buyers in class $k$ are considered before the buyers in class $k+1$, $k = 1, \ldots, K-1$. The buyers in the same class are considered in random order.
In the following, we let $\IE(\vec{q}, \vec{p})$ denote such a generalized IE strategy, where $\vec{q} = (q_1, \ldots, q_K)$ is the assignment probability vector and $\vec{p} = (p_1, \ldots, p_K)$ is the pricing probability vector.

We proceed to calculate the expected revenue extracted by the generalized IE strategy $\IE(\vec{q}, \vec{p})$ from an undirected social network $G(V, E, w)$.
The expected revenue of $\IE(\vec{q}, \vec{p})$ from each loop $\{ i, i \}$ is $w_{ii} \sum_{k=1}^K q_k p_k (1-p_k)$. Specifically, for each $k$, buyer $i$ is included in the pricing class $k$ with probability $q_k$, in which case, the revenue extracted from $\{i, i\}$ is $ p_k (1-p_k) w_{ii}$.
The expected revenue of $\IE(\vec{p}, \vec{q})$ from each edge $\{ i, j \}$, $i < j$, is:
\[
  w_{ij} \sum_{k=1}^K q_k p_k (1-p_k)
  \left( q_k p_k + 2\sum_{\ell=1}^{k-1} q_\ell p_\ell \right)
\]
More specifically, for each class $k$, if both $i$, $j$ are included in the pricing class $k$, which happens with probability $q_k^2$, the revenue extracted from $\{i, j\}$ is $p_k^2 (1-p_k)w_{ij} $. Furthermore, for each pair $\ell$, $k$ of pricing classes, $1 \leq \ell < k \leq K$, if either $i$ is included in $\ell$ and $j$ is included in $k$ or the other way around, which happens with probability $2 q_\ell q_k$, the revenue extracted from $\{i, j\}$ is $p_\ell p_k (1-p_k)w_{ij} $.
Using linearity of expectation and setting $N = \sum_{i \in V} w_{ii}$ and $W = \sum_{i < j} w_{ij}$, we obtain that the expected revenue of $\IE(\vec{q}, \vec{p})$ is:
\[
 \RIE(\vec{q}, \vec{p}) = N \sum_{k=1}^K q_k p_k(1-p_k) +
    W \sum_{k=1}^K q_k p_k(1-p_k)
      \left(q_k p_k + 2 \sum_{\ell=1}^{k-1} q_\ell p_\ell \right)
\]
Since $R^\ast = (N+W)/4$ is an upper bound on the maximum revenue of $G$, the approximation ratio of $\IE(\vec{q}, \vec{p})$ is at least:
\begin{equation}\label{eq:ie-gen-ratio}
  \min\!\left\{ 4 \sum_{k=1}^K q_k p_k(1-p_k),\ \
         4 \sum_{k=1}^K q_k p_k(1-p_k)
      \left(q_k p_k + 2 \sum_{\ell=1}^{k-1} q_\ell p_\ell \right) \right\}
\end{equation}
We can now select the assignment probability vector $\vec{q}$ so that (\ref{eq:ie-gen-ratio}) is maximized. We note that with the pricing probability vector $\vec{p}$ fixed, this involves maximizing a quadratic function of $\vec{q}$ over linear constraints. Thus, we obtain the following:

\begin{theorem}\label{th:ie-gen-undir}
For any undirected social network $G$, the generalized IE strategy with $K = 6$ pricing classes and assignment probabilities $\vec{q} = (0.183, 0.075, 0.075, 0.175, 0.261, 0.231)$ approximates the maximum revenue of $G$ within a factor of $0.7032$.
\end{theorem}

We note that the approximation ratio can be improved to $0.706$ by considering more pricing classes. By the same approach, we show that for directed social networks, the approximation ratio of $\IE(\vec{q}, \vec{p})$ is at least half the quantity in (\ref{eq:ie-gen-ratio}). Therefore:

\begin{corollary}\label{cor:ie-gen-dir}
For any directed social network $G$, the generalized IE strategy with $K = 6$ pricing classes and assignment probabilities $\vec{q} = (0.183, 0.075, 0.075, 0.175, 0.261, 0.231)$ approximates the maximum revenue of $G$ within a factor of $0.3516$.
\end{corollary}

\begin{proof}
Similarly to the proof of Theorem~\ref{th:ie-gen-undir}, we calculate the expected (wrt the random partition of buyers into pricing classes and the random order of buyers in the pricing classes) revenue extracted by the generalized IE strategy $\IE(\vec{p}, \vec{q})$ from a directed social network $G(V, E, w)$. We recall that for directed social networks, we can ignore loops $(i, i)$.
The expected revenue of $\IE(\vec{p}, \vec{q})$ from each edge $(i, j)$ is:
\[
  w_{ij} \sum_{k=1}^K q_k p_k (1-p_k)
  \left( \frac{q_k p_k}{2} + \sum_{\ell=1}^{k-1} q_\ell p_\ell \right)
\]
More specifically, for each class $k$, if both $i$, $j$ are included in the pricing class $k$ and $i$ appears before $j$ in the random order of the buyers in $k$, which happens with probability $q_k^2/2$, the revenue extracted from each edge $(i, j)$ is $p_k^2 (1-p_k) w_{ij}$. Furthermore, for each pair $\ell$, $k$ of pricing classes, $1 \leq \ell < k \leq K$, if $i$ is included in $\ell$ and $j$ is included in $k$, which happens with probability $q_\ell q_k$, the revenue extracted from $(i, j)$ is $p_\ell p_k (1-p_k) w_{ij}$.

Using linearity of expectation and setting $W = \sum_{(i, j) \in E} w_{ij}$, we obtain that the expected revenue of $\IE(\vec{q}, \vec{p})$ is:
\[
 \RIE(\vec{q}, \vec{p}) = W \sum_{k=1}^K q_k p_k(1-p_k)
      \left(\frac{q_k p_k}{2} + \sum_{\ell=1}^{k-1} q_\ell p_\ell \right)
\]
Since $W/4$ is an upper bound on the maximum revenue of $G$, the approximation ratio of $\IE(\vec{q}, \vec{p})$ is at least:
\begin{equation}\label{eq:ie-gen-dir-ratio}
         4 \sum_{k=1}^K q_k p_k(1-p_k)
      \left(\frac{q_k p_k}{2} + \sum_{\ell=1}^{k-1} q_\ell p_\ell \right)\,,
\end{equation}
namely at least half of the approximation ratio in the undirected case.

Using $\vec{q} = (0.183, 0.075, 0.075, 0.175, 0.261, 0.231)$ in (\ref{eq:ie-gen-dir-ratio}), we obtain an approximation ratio of at least $0.3516$.
 \qed\end{proof}

\section{Influence-and-Exploit via Semidefinite Programming}
\label{s:sdp}

The main hurdle in obtaining better approximation guarantees for the maximum revenue problem is the loose upper bound of $(N+W)/4$ on the optimal revenue. We do not know how to obtain a stronger upper bound on the maximum revenue. However, in this section, we obtain a strong Semidefinite Programming (SDP) relaxation for the problem of computing the best IE strategy with any given pricing probability $p \in [1/2, 1)$.
Our approach exploits the resemblance between computing the best IE strategy and the problems of MAX-CUT (for undirected networks) and MAX-DICUT (for directed networks), and builds on the elegant approach of Goemans and Williamson \cite{GW95} and Feige and Goemans \cite{FG95}.
Solving the SDP relaxation and using randomized rounding, we obtain, in polynomial time, a good approximation to the best influence set for the given pricing probability $p$. Then, employing the bounds of Theorem~\ref{th:bestIE-undir} and Theorem~\ref{th:bestIE-dir}, we obtain strong approximation guarantees for the maximum revenue problem for both directed and undirected networks.

\smallskip\noindent{\bf Directed Social Networks.}
We start with the case of a directed social network $G(V, E, w)$, which is a bit simpler, because we can ignore loops $(i, i)$ without loss of generality.
We observe that for any given pricing probability $p \in [1/2, 1)$, the problem of computing the best IE strategy $\IE(A, p$) is equivalent to solving the following Quadratic Integer Program:
\begin{align}
 \max & \,\,\tfrac{p(1-p)}{4}\!\!
 \sum_{(i, j) \in E} w_{ij}
 \left( 1+\tfrac{p}{2} + (1-\tfrac{p}{2})y_0y_i -
 (1+\tfrac{p}{2})y_0y_j - (1-\tfrac{p}{2})y_iy_j\right) \tag{Q1}
\end{align}
\vskip-7mm\begin{align}
 \mbox{s.t.} & & y_{i} \in \{ -1, 1 \} & &
 \forall i \in V \union \{ 0 \} \notag
\end{align}
In (Q1), there is a variable $y_i$ for each buyer $i$ and an additional variable $y_0$ denoting the influence set. A buyer $i$ is assigned to the influence set $A$, if $y_i = y_0$, and to the exploit set, otherwise. For each edge $(i, j)$, $1+y_0y_i-y_0y_j-y_iy_j$ is $4$, if $y_i = y_0 = -y_j$ (i.e., if $i$ is assigned to the influence set and $j$ is assigned to the exploit set), and $0$, otherwise. Moreover, $\frac{p}{2}(1-y_0y_i-y_0y_j+y_iy_j)$ is $2p$, if $y_i = y_j = -y_0$ (i.e., if both $i$ and $j$ are assigned to the exploit set), and $0$, otherwise. Therefore, the contribution of each edge $(i, j)$ to the objective function of (Q1) is equal to the revenue extracted from $(i, j)$ by $\IE(A, p)$.

Following the approach of \cite{GW95,FG95}, we relax (Q1) to the following Semidefinite Program, where $v_i\cdot v_j$ denotes the inner product of vectors $v_i$ and $v_j$:
\begin{align}
 \max & \,\,\tfrac{p(1-p)}{4}\!\!
 \sum_{(i, j) \in E} w_{ij}
 \left(1+\tfrac{p}{2} + (1-\tfrac{p}{2})\,v_0\cdot v_i -
 (1+\tfrac{p}{2})\,v_0\cdot v_j - (1-\tfrac{p}{2})\,v_i\cdot v_j\right)
 \tag{S1}
\end{align}
\vskip-7mm\begin{align}
 \mbox{s.t.}
 & & v_i\cdot v_j + v_0\cdot v_i + v_0\cdot v_j \geq -1 \notag \\
 & & \hskip-5cm v_i\cdot v_j - v_0\cdot v_i - v_0\cdot v_j \geq -1 \notag \\
 & & -v_i\cdot v_j - v_0\cdot v_i + v_0\cdot v_j \geq -1 \notag \\
 & & -v_i\cdot v_j + v_0\cdot v_i - v_0\cdot v_j \geq -1 \notag \\
 & & v_i \cdot v_i = 1,\ \ \ \ v_i \in \reals^{n+1}\ \ \ \ \ \ & & \forall i \in V \union \{ 0 \} \notag
\end{align}
We observe that any feasible solution to (Q1) can be translated into a feasible solution to (S1) by setting $v_i = v_0$, if $y_i = y_0$, and $v_i = -v_0$, otherwise. An optimal solution to (S1) can be computed within any precision $\eps$ in time polynomial in $n$ and in $\ln\frac{1}{\eps}$ (see e.g. \cite{Alizadeh}).

Given a directed social network $G(V, E, w)$, a pricing probability $p$, and a parameter $\gamma \in [0, 1]$, the algorithm $\SDP(p, \gamma)$ first computes an optimal solution $v_0, v_1, \ldots, v_n$ to (S1). Then, following \cite{FG95}, the algorithm maps each vector $v_i$ to a rotated vector $v'_i$ which is coplanar with $v_0$ and $v_i$, lies on the same side of $v_0$ as $v_i$, and forms an angle with $v_0$ equal to
\[
 f_\gamma(\theta_i) =
  (1-\gamma)\theta_i + \gamma \pi (1- \cos \theta_i)/2\,,
\]
where $\pi = 3.14\ldots$ and $\theta_i = \arccos(v_0\cdot v_i)$ is the angle of $v_0$ and $v_i$. Finally, the algorithm computes a random vector $r$ uniformly distributed on the unit $(n+1)$-sphere, and assigns each buyer $i$ to the influence set $A$, if $\sgn(v'_i\cdot r) = \sgn(v_0\cdot r)$, and to the exploit set $V \setminus A$, otherwise%
\footnote{Let $\theta'_i = \arccos(v_0\cdot v'_i)$ be the angle of $v_0$ and a rotated vector $v'_i$. To provide some intuition behind the rotation step, we note that $\theta'_i < \theta_i$, if $\theta_i \in (0, \pi/2)$, and $\theta'_i > \theta_i$, if $\theta_i \in (\pi/2, \pi)$. Therefore, applying rotation to $v_i$, the algorithm increases the probability of assigning $i$ to the influence set, if $\theta_i \in (0, \pi/2)$, and the probability of assigning $i$ to the exploit set, if $\theta_i \in (\pi/2, \pi)$. The strength of the rotation's effect depends on the value of $\gamma$ and on the value of $\theta_i$.},
where $\sgn(x) = 1$, if $x \geq 0$, and $-1$, otherwise. We next show that:

\begin{theorem}\label{th:sdp-dir}
For any directed social network $G$, $\SDP(2/3, 0.722)$ approximates the maximum revenue extracted from $G$ by the best IE strategy with pricing probability $2/3$ within a factor of $0.9064$.
\end{theorem}

\begin{proof}
In the following, we let $v_0, v_1, \ldots, v_n$ be an optimal solution to (S1), let $\theta_{ij} = \arccos(v_i \cdot v_j)$ be the angle of any two vectors $v_i$ and $v_j$, and let $\theta_i = \arccos(v_0\cdot v_i)$ be the angle of $v_0$ and any vector $v_i$. Similarly, we let $\theta'_{ij} = \arccos(v'_i \cdot v'_j)$ be the angle of any two rotated vectors $v'_i$ and $v'_j$, and let $\theta'_i = \arccos(v_0\cdot v'_i)$ be the angle of $v_0$ and any rotated vector $v'_i$.
We first calculate the expected revenue extracted from each edge $(i, j) \in E$ by the IE strategy of $\SDP(p, \gamma)$.

\begin{lemma}\label{l:dir-edge-exp-rev}
The IE strategy of $\SDP(p, \gamma)$ extracts from each edge $(i, j)$ an expected revenue of:
\begin{equation}\label{eq:dir-edge-exp-rev}
 w_{ij}\,p(1-p)\,
 \frac{(1-\tfrac{p}{2})\,\theta'_{ij} - (1-\tfrac{p}{2})\,\theta'_i
 + (1+\tfrac{p}{2})\,\theta'_j}{2\pi}
\end{equation}
\end{lemma}

\begin{proof}
%
We first define the following mutually disjoint events:
\begin{alignat*}{2}
& B^{ij}\ & :\ \sgn(v'_i\cdot r) = \sgn(v'_j\cdot r) = \sgn(v_0\cdot r) \\
& B^i_j\ & :\ \sgn(v'_i\cdot r) = \sgn(v_0\cdot r) \neq \sgn(v'_j\cdot r) \\
& B^j_i\ & :\ \sgn(v'_j\cdot r) = \sgn(v_0\cdot r) \neq \sgn(v'_i\cdot r) \\
& B_{ij}\ & :\ \sgn(v'_i\cdot r) = \sgn(v'_j\cdot r) \neq \sgn(v_0\cdot r)
\end{alignat*}
Namely, $B^{ij}$ (resp. $B_{ij}$) is the event that both $i$ and $j$ are assigned to the influence set $A$ (resp. to the exploit set $V \setminus A$), and $B^i_j$ (resp. $B^j_i$) is the event that $i$ (resp. $j$) is assigned to the influence set $A$ and $j$ (resp. $i$) is assigned to the exploit set $V \setminus A$. Also, we let $\Prob[B]$ denote the probability of any event $B$. Then, the expected revenue extracted from each edge $(i, j)$ is:
\begin{equation}\label{eq:dir-exp}
 w_{ij}\,p (1-p) \left(\Prob[B^i_j] + \tfrac{p}{2}\,\Prob[B_{ij}]\right)
\end{equation}

To calculate $\Prob[B^i_j]$ and $\Prob[B_{ij}]$, we use that if $r$ is a vector uniformly distributed on the unit sphere, for any vectors $v_i$, $v_j$ on the unit sphere,
%
\( \Prob[\sgn(v_i\cdot r) \neq \sgn(v_j\cdot r)] =
    \theta_{ij} / \pi \) \cite[Lemma~3.2]{GW95}.
%
For $\Prob[B^i_j]$, we calculate the probability of the event $B^i_j \union B^j_i$ that $i$ and $j$ are in different sets, of the event $B^i_j \union B^{ij}$ that $i$ is in the influence set, and of the event $B^j_i \union B^{ij}$ that $j$ is in the influence set.
\begin{alignat}{2}
 \Prob[B^i_j] + \Prob[B^j_i] &= \Prob[B^i_j \union B^j_i] & =
 \Prob[\sgn(v'_i\cdot r) \neq \sgn(v'_j\cdot r)] & = \theta'_{ij} / \pi
 \label{eq:dir-c1} \\
 \Prob[B^i_j] + \Prob[B^{ij}] & = \Prob[B^i_j \union B^{ij}] & =
 \Prob[\sgn(v'_i\cdot r) = \sgn(v_0\cdot r)] & = 1- \theta'_i / \pi
 \label{eq:dir-c2} \\
 \Prob[B^j_i] + \Prob[B^{ij}] & = \Prob[B^j_i \union B^{ij}] & =
 \Prob[\sgn(v'_j\cdot r) = \sgn(v_0\cdot r)] & = 1-\theta'_j / \pi
 \label{eq:dir-c3}
\end{alignat}
Subtracting (\ref{eq:dir-c3}) from (\ref{eq:dir-c1}) plus (\ref{eq:dir-c2}), we obtain that:
\begin{equation}\label{eq:dir-prob1}
 \Prob[B^i_j] = \tfrac{1}{2\pi}(\theta'_{ij} - \theta'_i + \theta'_j)
\end{equation}

For $\Prob[B_{ij}]$, we also need the probability of the event $B^j_i \union B_{ij}$ that $i$ is in the exploit set, and of the event $B^i_j \union B_{ij}$ that $j$ is in the exploit set.
\begin{alignat}{2}
 \Prob[B^j_i] + \Prob[B_{ij}] & = \Prob[B^j_i \union B_{ij}] & =
 \Prob[\sgn(v'_i\cdot r) \neq \sgn(v_0\cdot r)] & = \theta'_i / \pi
 \label{eq:dir-c4} \\
 \Prob[B^i_j] + \Prob[B_{ij}] & = \Prob[B^i_j \union B_{ij}] & =
 \Prob[\sgn(v'_j\cdot r) \neq \sgn(v_0\cdot r)] & = \theta'_j / \pi
 \label{eq:dir-c5}
\end{alignat}
Subtracting (\ref{eq:dir-c1}) from (\ref{eq:dir-c4}) plus (\ref{eq:dir-c5}), we obtain that:
\begin{equation}\label{eq:dir-prob2}
 \Prob[B_{ij}] = \tfrac{1}{2\pi}(-\theta'_{ij} + \theta'_i + \theta'_j)
\end{equation}

Substituting (\ref{eq:dir-prob1}) and (\ref{eq:dir-prob2}) in (\ref{eq:dir-exp}), we obtain (\ref{eq:dir-edge-exp-rev}), and conclude the proof of the lemma.
 \qed\end{proof}

Since (S1) is a relaxation of the problem of computing the best IE strategy with pricing probability $p$, the revenue of an optimal $\IE(A, p)$ strategy is at most:
\begin{equation}\label{eq:dir-opt}
 \tfrac{p(1-p)}{4}
 \sum_{(i, j) \in E} w_{ij}
 \left( 1+\tfrac{p}{2} + (1-\tfrac{p}{2})\cos\theta_i -
 (1+\tfrac{p}{2})\cos\theta_j - (1-\tfrac{p}{2})\cos\theta_{ij}\right)
\end{equation}
On the other hand, by Lemma~\ref{l:dir-edge-exp-rev} and linearity of expectation, the IE strategy of $\SDP(p, \gamma)$ generates an expected revenue of:
\begin{equation}\label{eq:dir-exp-rev}
 \tfrac{p(1-p)}{2\pi} \sum_{(i, j) \in E} w_{ij}
 \left( (1-\tfrac{p}{2})\,\theta'_{ij} - (1-\tfrac{p}{2})\,\theta'_i
 + (1+\tfrac{p}{2})\,\theta'_j \right)
\end{equation}
We recall that for each $i$, $\theta'_i = f_\gamma(\theta_i)$. Moreover, in \cite[Section~4]{FG95}, it is shown that for each $i$, $j$,
\[ \theta'_{ij} =
   g_\gamma(\theta_{ij}, \theta_i, \theta_j) =
   \arccos\!\left(
   \cos f_\gamma(\theta_i)\,\cos f_\gamma(\theta_j) +
   \frac{\cos\theta_{ij} - \cos\theta_i\,\cos\theta_j}
   {\sin\theta_i\,\sin\theta_j}\,
   \sin f_\gamma(\theta_i)\,\sin f_\gamma(\theta_j) \right)
\]
The approximation ratio of $\SDP(p, \gamma)$ is derived as the minimum ratio of any pair of terms in (\ref{eq:dir-exp-rev}) and (\ref{eq:dir-opt})  corresponding to the same edge $(i, j)$. Thus, the approximation ratio of $\SDP(p, \gamma)$ is:
\begin{align*}
 \rho(p, \gamma) & =
 \frac{2}{\pi} \min_{0 \leq x, y, z \leq \pi}
 \frac{(1-\tfrac{p}{2})\,g_\gamma(x, y, z)
 - (1-\tfrac{p}{2})f_\gamma(y) + (1+\tfrac{p}{2})f_\gamma(z)}
 {1+\tfrac{p}{2} + (1-\tfrac{p}{2})\cos y -
 (1+\tfrac{p}{2})\cos z - (1-\tfrac{p}{2})\cos x}
\end{align*}
\vskip-7mm\begin{alignat*}{2}
 \mbox{s.t.} &\ \ \ \ &
  \cos x + \cos y + \cos z \geq -1 \\
 & & \cos x - \cos y - \cos z \geq -1 \\
 & & - \cos x - \cos y + \cos z \geq -1 \\
 & & - \cos x + \cos y - \cos z \geq -1
\end{alignat*}
It can be shown numerically, that $\rho(2/3, 0.722) \geq 0.9064$.
 \qed\end{proof}

Combining Theorem~\ref{th:sdp-dir} and Theorem~\ref{th:bestIE-dir}, we conclude that:

\begin{theorem}\label{th:approx-dir}
For any directed social network $G$, the IE strategy computed by $\SDP(2/3, 0.722)$ approximates the maximum revenue of $G$ within a factor of $0.5011$.
\end{theorem}

\begin{figure}[t]
\begin{minipage}{0.45\textwidth}
\includegraphics[width=\textwidth]{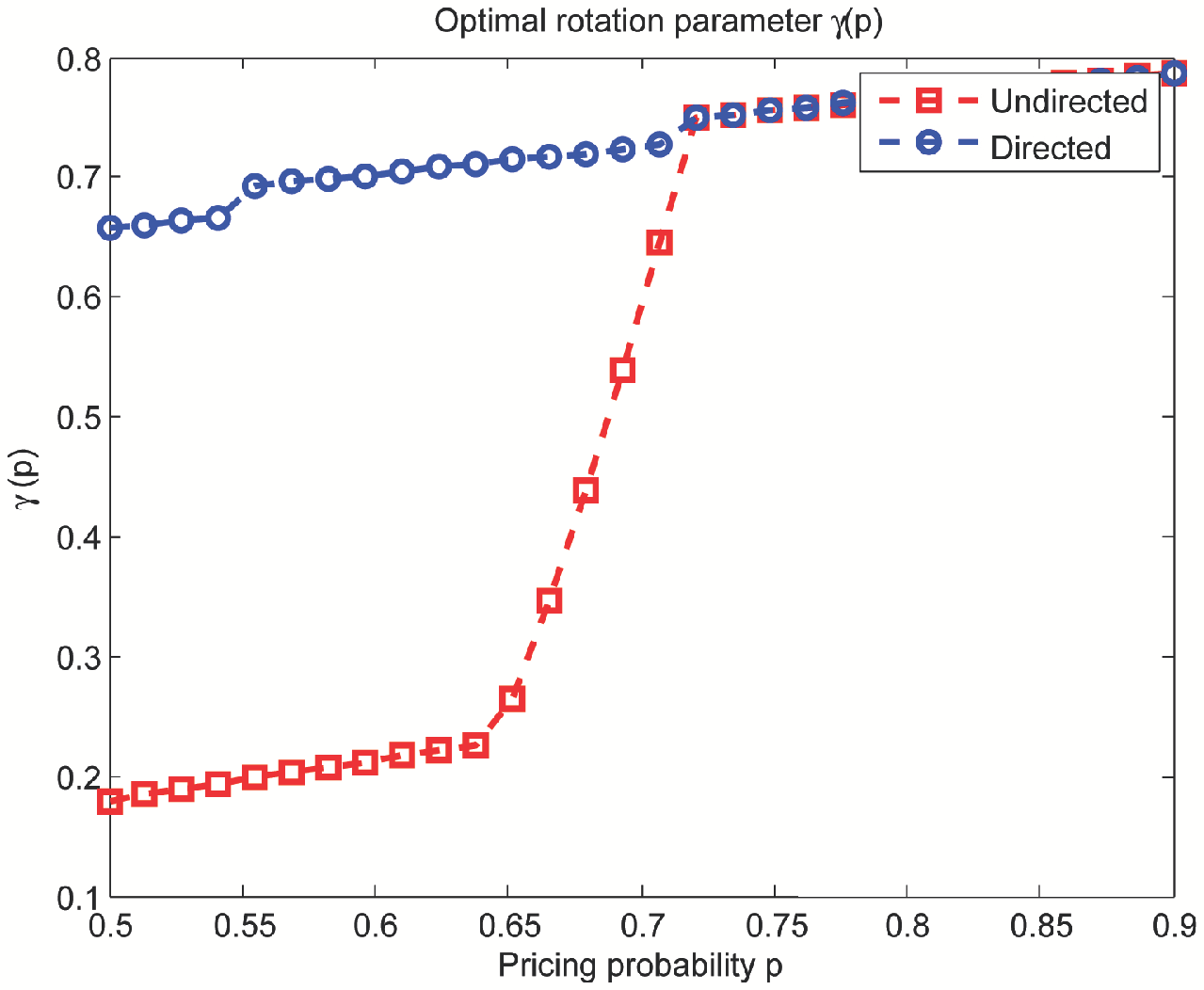}
\end{minipage}\hfill%
\begin{minipage}{0.45\textwidth}
\includegraphics[width=\textwidth]{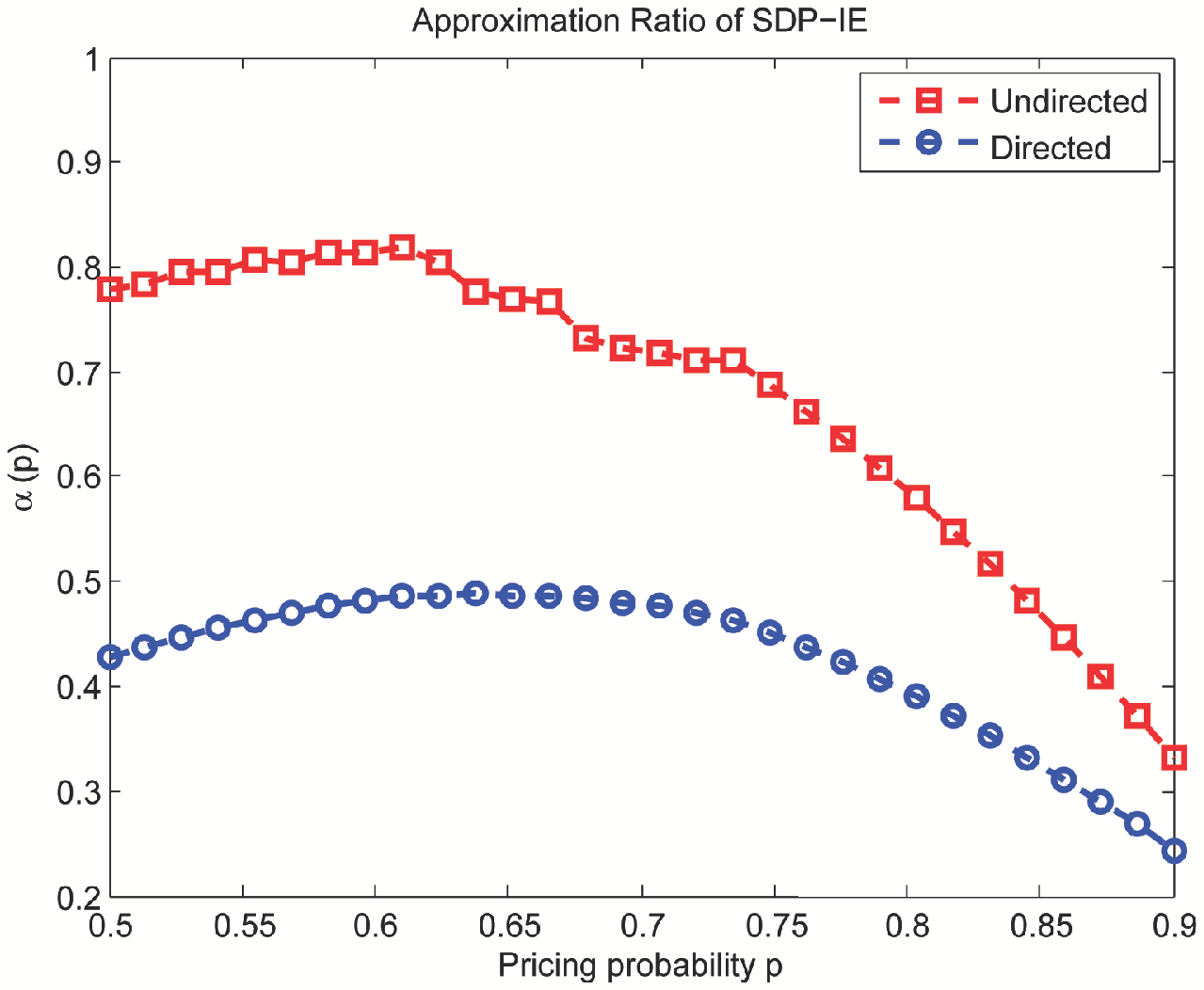}
\end{minipage}\\[1pt]
\begin{minipage}{0.45\textwidth}
\includegraphics[width=\textwidth]{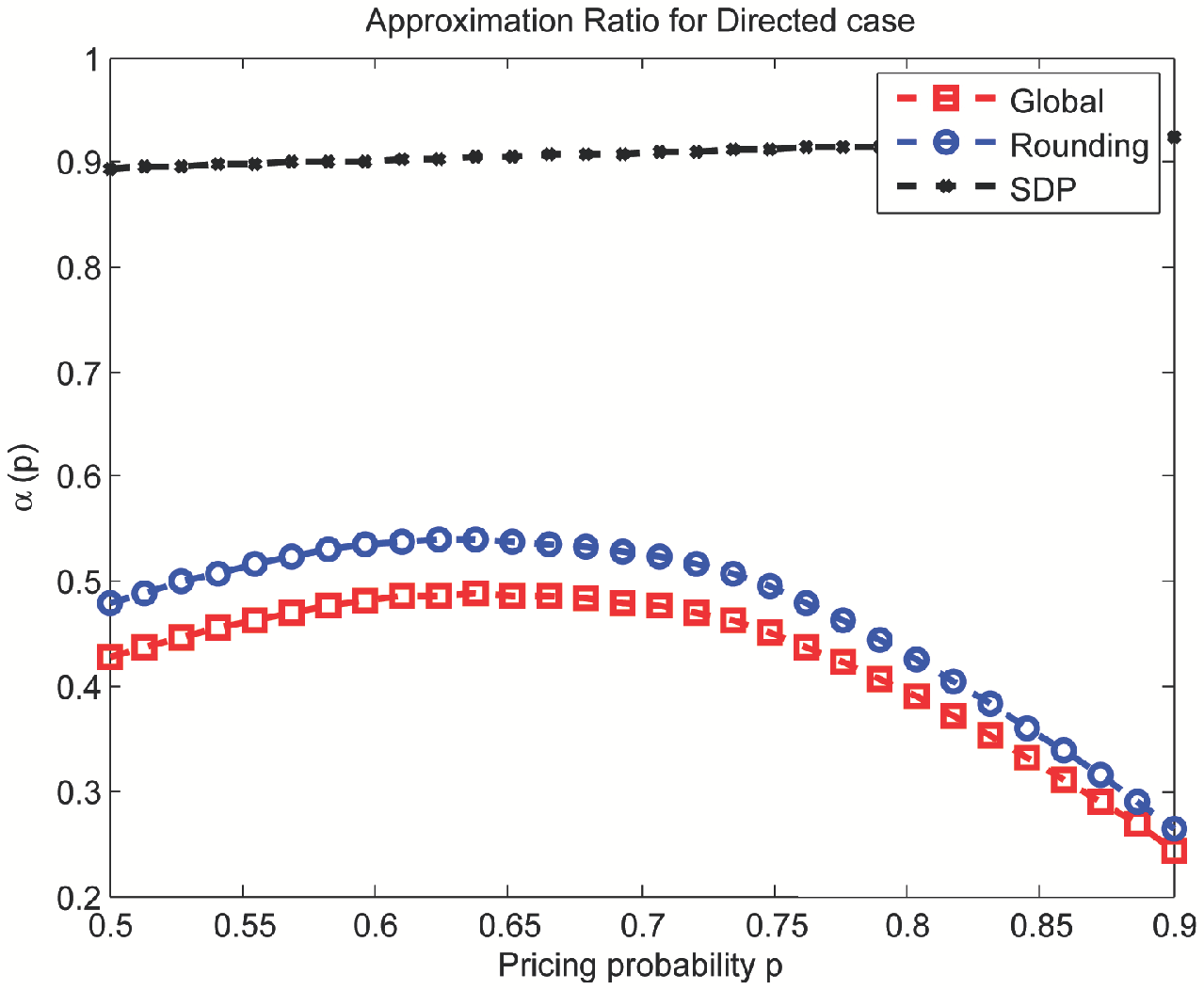}
\end{minipage}\hfill%
\begin{minipage}{0.45\textwidth}
\includegraphics[width=\textwidth]{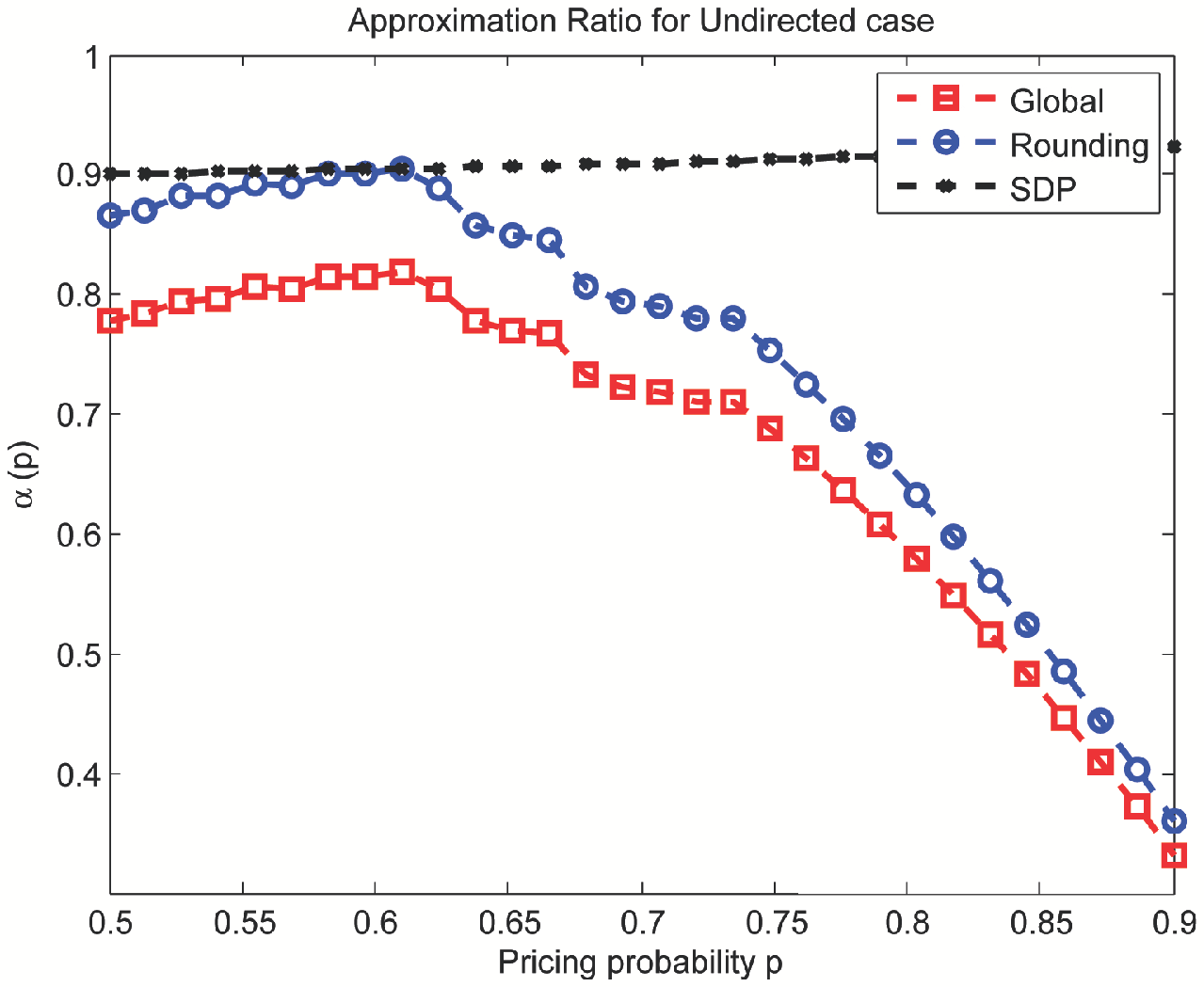}
\end{minipage}
\caption{\label{fig:approx}The approximation ratio of $\SDP(p, \gamma)$ for the revenue of the best IE strategy and for the maximum revenue, as a function of the pricing probability $p$.
The upper left plot shows the best choice of the rotation parameter $\gamma$, as a function of $p$. The blue curve (with circles) shows the best choice of $\gamma$ for directed social networks and the red curve (with squares) for undirected networks. In both cases, the best choice of $\gamma$ increases with $p$.
The upper right plot shows the approximation ratio of $\SDP(p, \gamma)$ for the maximum revenue for directed (blue curve, with circles) and undirected (red curve, with squares) networks.
The lower plots show the approximation ratio of $\SDP(p, \gamma)$ for directed (left plot) and undirected (right plot) networks, as a function of $p$.
In each plot, the upper curve (in black) shows the approximation ratio of $\SDP(p, \gamma)$ for the revenue of the best IE strategy, which increases slowly with $p$.
The blue curve (that with circles) shows the guarantee of Theorem~\ref{th:bestIE-dir} and Theorem~\ref{th:bestIE-undir} on the fraction of the maximum revenue extracted by the best IE strategy.
The red curve (that with squares) shows the approximation ratio of $\SDP(p, \gamma)$ for the maximum revenue.}
\end{figure}

\smallskip\noindent{\bf Undirected Social Networks.}
We apply the same approach to an undirected network $G(V, E, w)$.
For any given pricing probability $p \in [1/2, 1)$, the problem of computing the best IE strategy $\IE(A, p$) for $G$ is equivalent to solving the following Quadratic Integer Program:
\begin{align}
 \max & \,\,
 \tfrac{p(1-p)}{2}
 \sum_{i \in V} w_{ii}\,(1-y_0y_i) +
 \tfrac{p(1-p)}{4}
 \sum_{i < j} w_{ij}
 \left( 2+p - p y_0y_i - p y_0y_j - (2-p)y_iy_j\right) \tag{Q2}
\end{align}
\vskip-7mm\begin{align}
 \mbox{s.t.} & & y_{i} \in \{ -1, 1 \} & &
 \forall i \in V \union \{ 0 \} \notag
\end{align}
In (Q2), there is a variable $y_i$ for each buyer $i$ and an additional variable $y_0$ denoting the influence set. A buyer $i$ is assigned to the influence set $A$, if $y_i = y_0$, and to the exploit set, otherwise.
For each loop $\{ i, i \}$, $1-y_0y_i$ is $2$, if $i$ is assigned to the exploit set, and $0$, otherwise.
For each edge $\{ i, j \}$, $i < j$, $2-2y_iy_j$ is $4$, if $i$ and $j$ are assigned to different sets, and $0$, otherwise.
Also, $p(1-y_0y_i-y_0y_j+y_iy_j)$ is $4p$, if both $i$ and $j$ are assigned to the exploit set, and $0$, otherwise. Therefore, the contribution of each loop $\{i, i\}$ and each edge $\{i, j\}$, $i < j$, to the objective function of (Q2) is equal to the revenue extracted from them by $\IE(A, p)$.
The next step is to relax (Q1) to the following Semidefinite Program:
\begin{align*}
 \max & \,\,
 \tfrac{p(1-p)}{2}
 \sum_{i \in V} w_{ii}\,(1-v_0\cdot v_i) +
 \tfrac{p(1-p)}{4}
 \sum_{i < j} w_{ij}
 \left( 2+p - p\,v_0\cdot v_i - p\,v_0\cdot v_j - (2-p)\,v_i\cdot v_j\right)
\end{align*}
\vskip-7mm\begin{align}
 \mbox{s.t.}
 & & v_i\cdot v_j + v_0\cdot v_i + v_0\cdot v_j \geq -1 \tag{S2} \\
 & & \hskip-5cm v_i\cdot v_j - v_0\cdot v_i - v_0\cdot v_j \geq -1 \notag \\
 & & -v_i\cdot v_j - v_0\cdot v_i + v_0\cdot v_j \geq -1 \notag \\
 & & -v_i\cdot v_j + v_0\cdot v_i - v_0\cdot v_j \geq -1 \notag \\
 & & v_i \cdot v_i = 1,\ \ \ \ v_i \in \reals^{n+1}\ \ \ \ \ \ & & \forall i \in V \union \{ 0 \} \notag
\end{align}

The algorithm is the same as the algorithm for directed networks. Specifically, given an undirected social network $G(V, E, w)$, a pricing probability $p$, and a parameter $\gamma \in [0, 1]$, the algorithm $\SDP(p, \gamma)$ first computes an optimal solution $v_0, v_1, \ldots, v_n$ to (S2). Then, it maps each vector $v_i$ to a rotated vector $v'_i$ which is coplanar with $v_0$ and $v_i$, lies on the same side of $v_0$ as $v_i$, and forms an angle $f_\gamma(\theta_i)$ with $v_0$, where $\theta_i = \arccos(v_0\cdot v_i)$. Finally, the algorithm computes a random vector $r$ uniformly distributed on the unit $(n+1)$-sphere, and assigns each buyer $i$ to the influence set $A$, if $\sgn(v'_i\cdot r) = \sgn(v_0\cdot r)$, and to the exploit set $V \setminus A$, otherwise. We prove that:

\begin{theorem}\label{th:sdp-undir}
For any undirected network $G$, $\SDP(0.586, 0.209)$ approximates the maximum revenue extracted from $G$ by the best IE strategy with pricing probability $0.586$ within a factor of $0.9032$.
\end{theorem}

\begin{proof}
We employ the same approach, techniques, and notation as in the proof of Theorem~\ref{th:sdp-dir}. The expected revenue extracted from each loop $\{i, i\}$ is $w_{ii}\,p(1-p)$ times the probability that $i$ is in the exploit set, which is equal to
\( \Prob[\sgn(v'_i\cdot r) \neq \sgn(v_0\cdot r)] = \theta'_i/\pi \).
Therefore, the algorithm extracts an expected revenue of
\( w_{ii}\,p(1-p)\,\theta'_i / \pi \)
from each loop $\{i, i\}$. Next, we calculate the expected revenue extracted from each (undirected) edge $\{i, j\}$, $i < j$, by the IE strategy of $\SDP(p, \gamma)$.

\begin{lemma}\label{l:undir-edge-exp-rev}
$\SDP(p, \gamma)$ extracts from each edge $\{i, j\}$, $i < j$, an expected revenue of:
\[
 w_{ij}\,p(1-p)\,
 \frac{(2-p)\,\theta'_{ij} + p\,\theta'_i + p\,\theta'_j}
 {2\pi}
\]
\end{lemma}

\begin{proof}
Let the events $B^i_j$, $B^j_i$, and $B_{ij}$ be defined as in the proof of Lemma~\ref{l:dir-edge-exp-rev}. In particular, $B^i_j \union B^j_i$ is the event that $i$ and $j$ are in different sets, and $B_{ij}$ is the event that both $i$ and $j$ are in the exploit set. Thus, the expected revenue extracted from edge $\{i, j\}$ is:
\begin{equation}\label{eq:undir-exp}
 w_{ij}\,p (1-p) \left(\Prob[B^i_j \union B^j_i] + p\,\Prob[B_{ij}]\right)
\end{equation}
In the proof of Lemma~\ref{l:dir-edge-exp-rev}, in (\ref{eq:dir-c1}) and (\ref{eq:dir-prob2}) respectively, we show that
\( \Prob[B^i_j \union B^j_i] = \theta'_{ij} / \pi \),
and that
\( \Prob[B_{ij}] = (-\theta'_{ij} + \theta'_i + \theta'_j)/(2\pi) \).
Substituting these in (\ref{eq:undir-exp}), we obtain the lemma.
 \qed\end{proof}

Therefore, by linearity of expectation, the expected revenue of $\SDP(p, \gamma)$ is:
\begin{equation}\label{eq:undir-exp-rev}
 \tfrac{p(1-p)}{\pi} \sum_{i \in V} w_{ii}\,\theta'_i +
 \tfrac{p(1-p)}{2\pi} \sum_{i < j} w_{ij}
 \left( (2-p)\,\theta'_{ij} + p\,\theta'_i + p\,\theta'_j \right)\,,
\end{equation}
where $\theta'_i = f_\gamma(\theta_i)$, for each $i \in V$, and $\theta'_{ij} = g_\gamma(\theta_{ij}, \theta_i, \theta_j)$, for each $i, j \in V$.

On the other hand, since (S2) relaxes the problem of computing the best IE strategy with pricing probability $p$, the revenue of the best $\IE(A, p)$ strategy is at most:
\begin{equation}\label{eq:undir-opt}
 \tfrac{p(1-p)}{2} \sum_{i \in V} w_{ii}(1-\cos\theta_i) +
 \tfrac{p(1-p)}{4}
 \sum_{i < j} w_{ij} \left(2+p - p\cos\theta_i - p\cos\theta_j
  -(2-p)\cos\theta_{ij}\right)
\end{equation}

The approximation ratio of $\SDP(p, \gamma)$ is derived as the minimum ratio of any pair of terms in (\ref{eq:undir-exp-rev}) and (\ref{eq:undir-opt}) corresponding either to the same loop $\{i, i\}$ or to the same edge $\{i, j\}$, $i < j$. Therefore, the approximation ratio of $\SDP(p, \gamma)$ for undirected social networks is the minimum of $\rho_1(\gamma)$ and $\rho_2(p, \gamma)$, where:
\begin{align*}
 \rho_1(\gamma) & = \frac{2}{\pi}\min_{0 \leq x \leq \pi}
            \frac{f_\gamma(x)}{1-\cos x}
\hskip0.5cm\mbox{and}\\
 \rho_2(p, \gamma) & =
 \frac{2}{\pi} \min_{0 \leq x, y, z \leq \pi}
 \frac{(2-p)\,g_\gamma(x, y, z) + p f_\gamma(y) + p f_\gamma(z)}
 {2+p - p \cos y - p \cos z - (2-p)\cos x}
\end{align*}
\vskip-7mm\begin{alignat*}{2}
 \mbox{s.t.} &\ \ \ \ &
  \cos x + \cos y + \cos z \geq -1 \\
 & & \cos x - \cos y - \cos z \geq -1 \\
 & & - \cos x - \cos y + \cos z \geq -1 \\
 & & - \cos x + \cos y - \cos z \geq -1
\end{alignat*}
It can be shown numerically, that $\rho_1(0.209) \geq 0.9035$ and that
$\rho_2(0.586, 0.209) \geq 0.9032$.
 \qed\end{proof}

Combining Theorem~\ref{th:sdp-undir} and Theorem~\ref{th:bestIE-undir}, we conclude that:

\begin{theorem}\label{th:approx-undir}
For any undirected social network $G$, the IE strategy computed by $\SDP(0.586, 0.209)$ approximates the maximum revenue of $G$ within a factor of $0.8229$.
\end{theorem}

\noindent{\em Remark.}
We can use $\rho(p, \gamma)$ and $\min\{\rho_1(\gamma), \rho_2(p, \gamma)\}$, and compute the approximation ratio of $\SDP(p, \gamma)$ for the best IE strategy with any given pricing probability $p \in [1/2, 1)$. We note that $\rho_1(\gamma)$ is $\approx 0.87856$, for $\gamma = 0$ (see e.g. \cite[Lemma~3.5]{GW95}), and increases slowly with $\gamma$. Viewed as a function of $p$, the value of $\gamma$ maximizing $\rho(p, \gamma)$ and $\rho_2(p, \gamma)$ and the corresponding approximation ratio for the revenue of the best IE strategy increase slowly with $p$ (see also Fig~\ref{fig:approx} about the dependence of $\gamma$ and the approximation ratio as a function of $p$).
For example, for directed social networks, the approximation ratio of $\SDP(0.5, 0.653)$ (resp. $\SDP(0.52, 0.685)$ and $\SDP(0.52, 0.704)$) is $0.8942$ (resp. $0.8955$ and $0.9005$). For undirected networks, the ratio of $\SDP(0.5, 0.176)$ (resp. $\SDP(0.52, 0.183)$ and $\SDP(2/3, 0.425)$) is $0.899$ (resp. $0.9005$ and $0.907$). \qed

\bibliographystyle{plain}
\bibliography{pricing}

\appendix

\section{Appendix}

\subsection{Undirected Social Networks: An Example of a Suboptimal Ordering}
\label{app:cycle}

We consider an (undirected) simple cycle with $4$ nodes, numbered as they appear on the cycle, and unit weights on its edges. Proposition~\ref{pr:bipartite} shows that the optimal ordering is $(1, 3, 2, 4)$, the optimal pricing vector is $(1, 0.5, 1, 0.5)$, and the maximum revenue is $1$. On the other hand, if the nodes are ordered as they appear in the cycle, i.e., as in $(1, 2, 3, 4)$, the optimal pricing vector is $(1, \sqrt{2}/2,(1+\sqrt{2})/2, 0.5)$, and the resulting revenue is $0.7772$.

\subsection{On the Approximability of Maximum Revenue in the Uniform Additive Model}
\label{app:approximability}

We show a simple example where different acyclic subgraphs (equivalently, different vertex orderings) of the social network allow for a different fraction of their edge weight to be translated into revenue. To this end, we consider a simple directed network $G$ on $V = \{ u_1, u_2, u_3, u_4 \}$. $G$ contains an edge from each vertex $u_i$ to each vertex $u_j$ with $j > i$, that is $6$ edges in total. Formally, $E = \{ (u_i, u_j) : 1 \leq i < j \leq 4 \}$. The weight of each edge is $1$.

In ordering $\vec{\pi}_1 = (u_1, u_2, u_3, u_4)$, all edges go forward. So, $\vec{\pi}_1$ corresponds to an acyclic subgraph with edge weight $6$. The optimal pricing probabilities for $\vec{\pi}_1$ are $\vec{p}_1 = (1, 0.7474, 0.5715, 0.5)$ and extract a revenue of $R(\vec{\pi}_1, \vec{p}_1) = 1.1964$ from $G$. Thus, $\vec{\pi}_1$ allows for a revenue equal to $19.943\%$ of its edge weight.

Similarly, ordering $\vec{\pi}_2 = (u_1, u_3, u_2, u_4)$ corresponds to an acyclic subgraph with edge weight $5$. The optimal pricing probabilities for $\vec{\pi}_2$ are $\vec{p}_2 = (1, 0.625, 0.625, 0.5)$ and extract a revenue of $R(\vec{\pi}_2, \vec{p}_2) = 1.03125$. So, $\vec{\pi}_2$ allows for a revenue equal to $20.625\%$ of its edge weight.

Ordering $\vec{\pi}_3 = (u_2, u_1, u_3, u_4)$ also corresponds to an acyclic subgraph with edge weight $5$. The optimal pricing probabilities for $\vec{\pi}_3$ are $\vec{p}_3 = (1, 1, 0.5625, 0.5)$ and extract a revenue of $R(\vec{\pi}_3, \vec{p}_3) = 1.1328$. Thus, $\vec{\pi}_3$ allows for revenue equal to $22.656\%$ of its edge weight.
Also, the revenue extracted by $\IE(\{ u_1, u_2\}, 0.5147)$ is $1.0634$. Thus, $\vec{\pi}_3$ allows for an IE strategy extracting a revenue equal to $21.268\%$ of its edge weight.

$\IE(\{ u_1, u_2\}, 0.5147)$, for example, approximates the maximum revenue of $G$ within a factor of $\frac{1.0634}{1.1964} \approx 0.8888$. On the other hand, if we consider a random ordering of $u_1$ and $u_2$ and of $u_3$ and $u_4$, we obtain a vertex ordering $\vec{\pi}'$, which combined with $\vec{p}_1$, gives an expected revenue of $\approx 1.0306$. Hence, $(\vec{\pi}', \vec{p})$ approximates the maximum revenue of $G$ under $\vec{p}_1$ within a factor of $\frac{1.0306}{1.1964} \approx 0.8614$. On the other hand, $\vec{\pi}'$ defines an acyclic subgraph of $G$ which has an expected edge weight of $5$ and approximates the edge weight of the maximum acyclic subgraph of $G$ within a factor of $\frac{5}{6} \approx 0.8333$. \qed

\end{document}